%% file: master.tex
\documentclass[a4paper,UKenglish]{lipics-v2018}

\usepackage{microtype}
\usepackage{cdsty}
\usepackage{amsmath}
\usepackage{hyperref}

\ifdefined\LONGVERSION
  \relax
\else
 \newcommand{\LONGVERSION}[1]{#1}
 \newcommand{\SHORTVERSION}[1]{}
\fi
\newcommand{\LONGSHORT}[2]{\LONGVERSION{#1}\SHORTVERSION{#2}}

\newcommand{\tores}{\textsc{Tores}}

\newcommand{\ictx}{~\mathsf{ictx}}
\newcommand{\itype}{~\mathsf{itype}}

\usepackage[]{listings}
\lstdefinelanguage{Beluga}
{
  morekeywords={rec, fun, LF, inductive, coinductive, stratified,type, ctype,
    and, schema, some, let, in, case, of},
  keepspaces=true,
  sensitive,
  morecomment=[l]{\%},
  morecomment=[n]{\%\{}{\}\%},
  morestring=[b]"
}[keywords,comments,strings]

\newcommand{\beluga}{\textsc{Beluga}}

\newcommand{\sep}{\;|\;}
\DeclareMathOperator{\Rec}{\mathtt{Rec}}
\newcommand{\psplit}[3]{\ensuremath{\code{split}\,#1\,\code{as}\,\pair{#2}{#3}\,\code{in}\,}}

\newcommand{\many}[1]{\overrightarrow{#1}}
\newcommand{\pr}{\ensuremath{^\prime}}

\newcommand{\bnfas}{\mathrel{::=}}

\newcommand{\code}[1]{\texttt{#1}}
\newcommand{\mtt}[1]{\mathtt{#1}}

\newcommand{\oft}{\ensuremath{\mathrel{:}}} 
\newcommand{\chk}{\Leftarrow}
\newcommand{\syn}{\Rightarrow}
\newcommand{\gives}{\ensuremath{\mathrel{\vdash}}} 
\newcommand{\evto}{\ensuremath{\Downarrow}} 
\newcommand{\eval}[3]{\ensuremath{#1[#2] \evto #3}}
\newcommand{\evalapp}[4]{\ensuremath{#1 \cdot #2\,#3 \evto #4}}

\newcommand{\fall}[1]{\forall #1.\;}

\DeclareMathOperator{\nat}{\mathsf{nat}}

\DeclareMathOperator{\suc}{\mathsf{suc}}

\DeclareMathOperator{\Vect}{\mathsf{Vec}}
\DeclareMathOperator{\Red}{\mathsf{Red}}
\DeclareMathOperator{\Stream}{\mathsf{Stream}}

\newcommand{\id}{\textsf{id}}

\newcommand{\unit}{\langle\rangle}
\newcommand{\lam}[1]{\ensuremath{\lambda\,#1.\,}}
\newcommand{\ilam}[2]{\ensuremath{\lambda\,#1,#2.\,}}
\newcommand{\rec}[1]{\ensuremath{\mtt{rec}\,#1.\,}}
\newcommand{\corec}[1]{\ensuremath{\mtt{corec}\,#1.\,}}
\newcommand{\app}{\,}
\newcommand{\appv}{\ensuremath{\mathrel{\cdot}}}
\newcommand{\appval}[3]{\ensuremath{#1 \appv #2\,#3}}
\newcommand{\appout}[1]{\ensuremath{#1 \appv_{\out{\nu}}}}
\newcommand{\pair}[2]{\ensuremath{\langle #1, #2 \rangle}}

\newcommand{\inj}[1]{\code{in}_{#1}\,}
\newcommand{\out}[1]{\code{out}_{#1}\,}
\newcommand{\ann}{\ensuremath{{:}}}

\newcommand{\case}[1]{\ensuremath{\code{case}\,#1\,\code{of}\,}}
\newcommand{\casetwo}[5]{\ensuremath{\code{case}\,#1\,\code{of}\,\inj{1}\,#2 \mapsto #3 \sep \inj{2}\,#4 \mapsto #5}}
\newcommand{\pack}[2]{\ensuremath{\code{pack}\,(#1, #2)}}
\newcommand{\unpack}[3]{\ensuremath{\code{unpack}\,#1\,\code{as}\,(#2, #3)\,\code{in}\,}}

\newcommand{\refl}{\code{refl}}
\newcommand{\eqelim}[4]{\ensuremath{\code{eq}\,#1\, \code{with}\,(#3. #2\,\mapsto\,#4)\,}}
\newcommand{\eqelimfalse}[1]{\ensuremath{\code{eq\_abort}\,#1\,}}
\newcommand{\abort}{\code{abort}\,}

\newcommand{\ind}{\ensuremath{\code{ind}\,}}
\newcommand{\indnat}[4]{\ensuremath{\code{ind}\,#1\,(#2,\,#3.\,#4)}}

\newcommand{\type}{*}
\newcommand{\pik}[1]{\ensuremath{\Pi #1.\,}}

\newcommand{\iarr}[3]{\ensuremath{(#1);\,#2 \to #3}}

\newcommand{\sit}[2]{\ensuremath{\Sigma #1.\,#2}}
\newcommand{\sigt}[1]{\ensuremath{\Sigma #1.\,}}
\newcommand{\cross}{\ensuremath{\times}}
\newcommand{\tsum}{\ensuremath{+}}

\newcommand{\Unit}{\ensuremath{1}}

\newcommand{\Lam}[1]{\ensuremath{\Lambda\,#1.\,}}
\newcommand{\App}{\,}

\newcommand{\Recnat}[5]{\ensuremath{\Rec_{#5}\,(0 \mapsto #1 \sep \suc #2,\,#3 \mapsto #4)}}

\newcommand{\rect}[1]{\ensuremath{\mu #1.\,}}
\newcommand{\corect}[1]{\ensuremath{\nu #1.\,}}
\newcommand{\empt}{\cdot}
\newcommand{\sub}[2]{#1/#2}

\newcommand{\powset}[1]{\ensuremath{\mathcal{P}(#1)}}
\newcommand{\mc}[1]{\ensuremath{\mathcal{#1}}}

\DeclareMathOperator{\FV}{FV}

\newcommand{\unif}{\doteqdot}

\newcommand{\match}{\doteq}
\newcommand{\fail}{\#}
\newcommand{\gen}{\searrow}
\newcommand{\synth}{\searrow}

\newcommand{\rl}[1]{\code{#1}}

\newcommand{\eunit}{\rl{e-unit}}
\newcommand{\evar}{\rl{e-var}}
\newcommand{\elam}{\rl{e-lam}}
\newcommand{\erec}{\rl{e-rec}}
\newcommand{\ecorec}{\rl{e-corec}}
\newcommand{\eind}{\rl{e-ind}}
\newcommand{\eapp}{\rl{e-app}}
\newcommand{\epair}{\rl{e-pair}}
\newcommand{\esplit}{\rl{e-split}}

\newcommand{\einj}[1]{\rl{e-in}_{#1}}
\newcommand{\eout}[1]{\rl{e-out}_{#1}}
\newcommand{\ecorecout}{\rl{e-corec-out}_{\nu}}
\newcommand{\erefl}{\rl{e-refl}}
\newcommand{\eeq}{\rl{e-eq}}
\newcommand{\epack}{\rl{e-pack}}
\newcommand{\eunpack}{\rl{e-unpack}}

\newcommand{\ecase}[1]{\rl{e-case-in}_{#1}}
\newcommand{\eapplam}{\rl{e-app-lam}}
\newcommand{\eapprec}{\rl{e-app-rec}}
\newcommand{\eappind}[1]{\rl{e-app-ind}_{#1}}

\newcommand{\tunit}{\rl{t-unit}}
\newcommand{\tvar}{\rl{t-var}}
\newcommand{\tlam}{\rl{t-lam}}
\newcommand{\tapp}{\rl{t-app}}
\newcommand{\tpair}{\rl{t-pair}}
\newcommand{\tsplit}{\rl{t-split}}

\newcommand{\tinj}[1]{\rl{t-in}_{#1}}
\newcommand{\tout}[1]{\rl{t-out}_{#1}}
\newcommand{\tcase}{\rl{t-case}}

\newcommand{\trec}{\rl{t-rec}}
\newcommand{\tcorec}{\rl{t-corec}}
\newcommand{\trefl}{\rl{t-refl}}
\newcommand{\teq}{\rl{t-eq}}
\newcommand{\teqfalse}{\rl{t-eqfalse}}
\newcommand{\tpack}{\rl{t-pack}}
\newcommand{\tunpack}{\rl{t-unpack}}
\newcommand{\tind}{\rl{t-ind}}

\newcommand{\U}{\mathcal{U}}
\newcommand{\V}{\mathcal{V}}
\newcommand{\A}{\mathcal{A}}
\newcommand{\B}{\mathcal{B}}
\newcommand{\C}{\mathcal{C}}

\newcommand{\E}{\mathcal{E}}
\newcommand{\F}{\mathcal{F}}

\newcommand{\X}{\mathcal{X}}

\newcommand{\VAL}{\Omega}
\newcommand{\N}{\mathbb N}

\newcommand{\sem}[1]{\ensuremath{\llbracket #1 \rrbracket}}

\newcommand{\semarr}[3]{#1\boldsymbol{,}\;#2\boldsymbol{\to}#3}
\newcommand{\semlfp}{\boldsymbol{\mu}}
\newcommand{\semgfp}{\boldsymbol{\nu}}

\newcommand{\injop}[1]{\code{in}_{#1}^*\,}
\newcommand{\outop}[1]{\code{out}_{#1}^*\,}

\newcommand{\meet}{\bigwedge}
\newcommand{\join}{\bigvee}

\newcommand{\CROSS}{\mathop{\bm{\times}\kern-1.93ex\bm{\times}}}
\newcommand{\TSUM}{\mathop{\bm{+}\kern-1.93ex\bm{+}}}

\DeclareMathOperator{\REC}{\bold{Rec}}
\newcommand{\RECnat}[2]{\REC\,#1\,#2}

\newenvironment{figureone}[1]{%
  \def\deffigurecaption{#1}%
  \begin{figure}[ht!]%
  \begin{center}%
  \begin{minipage}{\columnwidth}%
  \vspace*{2ex}%
}{%
  \end{minipage}%
  \end{center}%
  \caption{\deffigurecaption}%
  \end{figure}%
}

\usepackage{proof}    
\usepackage{stmaryrd} 
\usepackage{paralist} 
\usepackage[authoryear]{natbib}
\usepackage{todonotes}
\usepackage{chngcntr} 
\usepackage{thmtools, thm-restate} 

\newtheorem{req}{Requirement}

\usepackage[capitalise]{cleveref}
\crefname{section}{Section}{Sections}
\crefname{definition}{Def.}{Definitions}
\crefname{lemma}{Lemma}{Lemmas}
\crefname{theorem}{Thm}{Theorems}
\crefname{req}{Req.}{Requirements}

\hideLIPIcs
\title{Index-Stratified Types (Extended Version)}
\funding{This research was funded by the Natural Science and Engineering Research Council (NSERC) Canada.}
\acknowledgements{We thank Andrew Cave for the idea of stratified types and for guiding the initial development.}
\SHORTVERSION{\relatedversion{A long version of this paper with the full technical appendix is available at \url{http://www.rohanjr.com/stratified.pdf}.}}

\author{Rohan Jacob-Rao}{Digital Asset, Sydney, Australia}{rohanjr@digitalasset.com}{}{}
\author{Brigitte Pientka}{McGill University, Montreal, Canada}{bpientka@cs.mcgill.ca}{}{}
\author{David Thibodeau}{McGill University, Montreal, Canada}{david.thibodeau@mail.mcgill.ca}{}{}

\authorrunning{R. Jacob-Rao and B. Pientka and D. Thibodeau} 

\Copyright{Rohan Jacob-Rao and Brigitte Pientka and David Thibodeau}

\subjclass{
  \\
}

\EventEditors{H\'{e}l\`{e}ne Kirchner}
\EventNoEds{1}
\EventLongTitle{3rd International Conference on Formal Structures for Computation and Deduction (FSCD 2018)}
\EventShortTitle{FSCD 2018}
\EventAcronym{FSCD}
\EventYear{2018}
\EventDate{July 9--12, 2018}
\EventLocation{Oxford, UK}
\EventLogo{}
\SeriesVolume{108}
\ArticleNo{19}

\nolinenumbers

\begin{document}

\maketitle

\begin{abstract}
We present \tores, a core language for encoding metatheoretic proofs.  The novel features we introduce
are well-founded Mendler-style (co)recursion over indexed
data types and a form of recursion over objects in the index
language to build new types.  The latter, which we
call \emph{index-stratified types}, are analogue to the concept of
large elimination in dependently typed languages.  These features
combined allow us to encode sophisticated case studies such as
normalization for lambda calculi and normalization by evaluation.  We
prove the soundness of \tores\/ as a programming and proof language
via the key theorems of subject reduction and termination.
\end{abstract}

\section{Introduction}
Recursion is a fundamental tool for writing useful programs in functional languages.
When viewed from a logical perspective via the Curry-Howard correspondence, well-founded recursion
corresponds to inductive reasoning. Dually, well-founded corecursion
corresponds to coinductive reasoning. However, concentrating only on
well-founded (co)recursive definitions is not sufficient to support
the encoding of meta-theoretic proofs. There are two missing ingredients:
1) To express fine-grained properties we often rely on first-order
logic which is analogous to  \emph{indexed types} in programming
languages. 2) Many common notions cannot be directly characterized by
well-founded (co)recursive definitions. An example is Girard's notion of
reducibility for functions: a term $M$ is reducible at type $A \rightarrow B$
if, for all terms $N$ that are reducible at type $A$, we have that $M~N$ is
reducible at type $B$.
This definition is well-founded because it is by structural recursion
on the type indices ($A$ and $B$), so we want to admit such definitions.


Our contribution in this paper is a core language called \tores\/ that features indexed types and
(co)inductive reasoning via well-founded (co)recursion.
The primary forms of types are \emph{indexed (co)recursive types}, over which we support reasoning via
Mendler-style (co)recursion.
Additionally, \tores\/ features \emph{index-stratified types}, which allow further definitions of
types via well-founded recursion over indices.
The main difference between the two forms is that (co)recursive types are more flexible, allowing
(co)induction, while stratified types only support unfolding based on their indices.
The combination of the two features is especially powerful for formalizing metatheory involving
logical relations.
This is partly because type definitions in \tores\/ do not require positivity, a condition used in
other systems to ensure termination and in turn logical consistency.
Despite this, we are able to prove termination of \tores\/ programs using a semantic interpretation
of types.

How to justify definitions that are recursively defined on
a given index in addition to well-founded (co)recursive definitions
has been explored in proof theory (see for example
\citet{Tiu:IJCAR12,Baelde:LICS12}). While this line of work is more
general, it is also more complex and further from
standard programming practice. In dependent type theories, large
eliminations achieve the same. Our approach,
grounded in the Curry-Howard isomorphism, provides a complementary
perspective on this problem where we balance expressiveness and ease
of programming with a compact metatheory. We believe this may be
an advantage when considering more sophisticated index languages and
reasoning techniques.

The combination of indexed (co)recursive types and stratified types is already used in the programming
and proof environment \beluga\/, where the index language is an extension of the logical framework
LF together with first-class contexts and substitutions
\citep{Nanevski:ICML05,Pientka:POPL08,Cave:POPL12}.
This allows elegant implementations of proofs using logical relations
\citep{Cave:LFMTP13,Cave:LFMTP15} and normalization by evaluation \citep{Cave:POPL12}.
\tores\/ can be seen as small kernel into which we elaborate total \beluga\/ programs, thereby
providing post-hoc justification of viewing \beluga~programs as (co)inductive proofs. 

\section{Index Language for \tores}
\label{sec:tores-index}
The design of \tores\/ is parametric over an index language.
Following \citet{Thibodeau:ICFP16} we stay as abstract as possible and state the general conditions the index language must satisfy.
Whenever we require inspection of the particular index language, 
namely the structure of stratified types and induction terms, we will draw attention to it. 

To illustrate the required structure for a concrete index language, we use 
natural numbers.
In practice, however, we can consider other index languages such as those of strings,
types \citep{ChH03Pha,Xi:POPL03}, or (contextual) LF \citep{Pientka:POPL08,Cave:POPL12}.
It is important to note that, for most of our design, we accommodate a general index language up to
the complexity of Contextual LF.
Thus we treat index types and \tores\/ kinds as dependently typed\LONGVERSION{, although we use natural numbers
in stratified types and induction terms}.

\LONGVERSION{The abstract requirements of our index language are listed throughout this section.
To summarize them here, they are: decidable type checking, decidable equality, standard substitution
principles, decidable unification as well as sound and complete matching.
Implicitly, we also require that each index type intended for use in stratified types and induction
terms should have a well-founded recursion scheme, i.e. an induction principle.
For an index language of Contextual LF, for example, the recursion scheme can be generated using a
covering set of index terms for each index type \citep{Pientka:TLCA15}.
This inductive structure is necessary to show decidability of type checking
(\cref{thm:dec-typ-chk}) and termination (\cref{thm:term}) of \tores.
}{}

\subsection{General Structure}

We refer to a term in the index language as an \emph{index term} $M$, which may have an \emph{index type} $U$.
In the case of natural numbers, there is a single index type $\nat$, and index terms are
built from $0$, $\suc$, and variables $u$ which must be declared in an \emph{index context} $\Delta$.
\[
\begin{array}{lll@{~\qquad~}lll}
\mbox{Index types} & U &:= \nat & \mbox{Index contexts} & \Delta &:= \empt \sep \Delta, u \ann U\\
\mbox{Index terms} & M &:= 0 \sep \suc M \sep u 
&
\mbox{Index substitutions} & \Theta & := \empt \sep \Theta,\sub{M}{u} \\
\end{array}
\]
\tores\/ relies on typing for index terms which we give for natural
numbers in \cref{fig:index}. 
The equality judgment for natural numbers is given simply by reflexivity (syntactic equality).
We also give typing for index substitutions, which supply an index term for each index variable in
the domain $\Delta$ and describe well-formed contexts. These definitions are generic.

\begin{figureone}{
  Index language structure
  \label{fig:index}
 }
\[
\begin{array}{c}
\multicolumn{1}{l}{
  \fbox{$\gives \Delta \ictx$}
  \quad
  \mbox{Index context $\Delta$ is well-formed}
  \qquad
    \fbox{$\Delta \gives U \itype$}
  \quad
  \mbox{Index type $U$ is well-kinded}
}\\[0.75em]
\infer{\gives \empt \ictx}{} \quad
\infer{\gives \Delta, u \ann U \ictx}{\gives \Delta \ictx & \Delta \gives U \itype}
\qquad\qquad\qquad\qquad\qquad\qquad\quad
\infer{\Delta \gives \nat \itype}{}
\\[0.75em]
\multicolumn{1}{l}{
  \fbox{$\Delta \gives M \oft U$}
  \quad
  \mbox{Index term $M$ has index type $U$ in index context $\Delta$}
}\\[0.75em]
\infer{\Delta \gives u \oft U}{u \ann U \in \Delta}
\quad
\infer{\Delta \gives 0 \oft \nat}{}
\quad
\infer{\Delta \gives \suc M \oft \nat}{\Delta \gives M \oft \nat}
\\[0.75em]
\multicolumn{1}{l}{
  \fbox{$\Delta \gives M = N$}
  \quad
  \mbox{Index term $M$ is equal to $N$}
}\\[0.75em]
\infer{\Delta \gives M = M}{}
\\[0.75em]
\multicolumn{1}{l}{
  \fbox{$\Delta\pr \gives \Theta \oft \Delta$}
  \quad
  \mbox{Index substitution $\Theta$ maps index variables from $\Delta$ to $\Delta\pr$}
}\\[0.75em]
\infer{\Delta\pr \gives \empt \oft \empt}{}
\qquad
\infer{\Delta\pr \gives \Theta,\sub{M}{u} \oft \Delta,u \ann U}
      {\Delta\pr \gives \Theta \oft \Delta &
       \Delta\pr \gives M \oft U[\Theta]}
\end{array}
\]
 \end{figureone}

We require that both typing and equality of index terms be decidable in order for type checking of
\tores\/ programs to be decidable.

\begin{req}
\label{req:dec-idx-typ-chk}
  Index type checking is decidable.
\end{req}

\begin{req}
\label{req:dec-idx-eq}
  Index equality is decidable.
\end{req}


We can lift the kinding, typing, equality and matching rules to \emph{spines} of index
terms and types generically.
We write $\empt$ and $(\empt)$ for the empty spines of terms and types respectively.
If $M_0$ is an index term and $\many{M}$ is a spine, then $M_0, \many{M}$ is a spine.
Similarly if $u_0 \ann U_0$ is an index type assignment and $(\many{u \ann U})$ is a type spine,
then $(u_0 \ann U_0, \many{u \ann U})$ is a type spine.
Spines are convenient for setting up the types and terms of \tores.
Unlike index substitutions $\Theta$ which are built
from right to left, spines are built from left to right.

\LONGVERSION{Below we define well-kinded spines of index types and well-typed spines of index
terms, which are generic to the particular index language.
\input{spine.tex}
}{}
\LONGVERSION{
\begin{lemma}
\label{lem:dec-idx-spine-chk}
  Type checking of index spines is decidable.
\end{lemma}
\begin{proof}
  Simply rely on decidable type checking of single index terms (\cref{req:dec-idx-typ-chk}).
\end{proof}
}

\LONGVERSION{\subsection{Substitutions}}{}
Throughout our development we use both a single index substitution operation $M[N/u]$ and a
simultaneous substitution operation $M[\Theta]$. For composition of
simultaneous substitutions we write $\Theta_1[\Theta_2]$.
\LONGVERSION{
\begin{definition}[Composition of index substitutions]
\label{def:isubst-comp}
Suppose $\Delta_1 \gives \Theta_1 \oft \Delta$ and $\Delta_2 \gives \Theta_2 \oft \Delta_1$.
Then $\Delta_2 \gives \Theta_1[\Theta_2] \oft \Delta$ where
  \begin{align*}
    (\empt)[\Theta_2] &= \Theta_2 \\
    (\Theta_1\pr,\sub{M}{u})[\Theta_2] &= \Theta_1\pr[\Theta_2],\sub{M[\Theta_2]}{u}
  \end{align*}
\end{definition}
}{}

\LONGVERSION{We rely on standard properties of single and simultaneous
substitutions which we summarize below. These say that substitutions preserve typing (\ref{req:isubst-typ-sing} and \ref{req:isubst-typ-sim})
and equality (\ref{req:isubst-eq}) and are associative (\ref{req:isubst-comp}).
}
\begin{req}[Index substitution principles]
  \label{req:isubst}$\;$\\[-1.5em]
  \counterwithin{enumi}{req}
  \begin{enumerate}
    \item\label{req:isubst-typ-sing}
      If $\Delta_1, u \ann U\pr, \Delta_2 \gives M \oft U$ and $\Delta_1 \gives N \oft U\pr$ then
      $\Delta_1, \Delta_2[\sub{N}{u}] \gives M[\sub{N}{u}] \oft U[\sub{N}{u}]$.
    \item\label{req:isubst-typ-sim}
      If $\Delta\pr \gives \Theta \oft \Delta$ and $\Delta \gives M \oft U$ then
      $\Delta\pr \gives M[\Theta] \oft U[\Theta]$.
    \item\label{req:isubst-eq}
      If $\Delta\pr \gives \Theta \oft \Delta$ and $\Delta \gives M = N$ then
      $\Delta\pr \gives M[\Theta] = N[\Theta]$.
    \item\label{req:isubst-comp}
      If $\Delta \gives M \oft U$ and
      $\Delta_1 \gives \Theta_1 \oft \Delta$ and
      $\Delta_2 \gives \Theta_2 \oft \Delta_1$ then
      $M[\Theta_1][\Theta_2] = M[\Theta_1[\Theta_2]]$.
  \end{enumerate}
\end{req}

\subsection{Unification and Matching}
\label{sec:index-unif}

Type checking of \tores\/ relies on a unification procedure to generate a \emph{most
general unifier (MGU)}.
A \emph{unifier} for index terms $M$ and $N$ in a context $\Delta$ is a substitution $\Theta$ which
transforms $M$ and $N$ into syntactically equal terms in another context $\Delta\pr$.
That is, $\Delta\pr \gives \Theta \oft \Delta$ and $\Delta\pr \gives M[\Theta] = N[\Theta]$.
$\Theta$ is ``most general'' if it does not make more commitments to variables than absolutely
necessary.
A unifying substitution $\Theta$ only makes sense together with its range $\Delta\pr$, so we usually
write them as a pair $(\Delta\pr \mid \Theta)$.
In general, there may be more than one MGU for a particular unification problem, or none at all.
However, we require here that each problem has at most one MGU up to $\alpha$-equivalence.
We write the generation of an MGU using the judgment $\Delta \gives M \unif N \gen P$, where $P$ is
either the MGU $(\Delta\pr \mid \Theta)$ if it exists or $\fail$ representing that unification
failed. To illustrate, we show the unification rules for natural
numbers. We write $\id_i$ for the identity substitution that maps
index variables from $\Delta_i$ to themselves.
\[
\begin{array}{l}
\fbox{$\Delta \gives M \unif N \gen P$}~~\mbox{Index terms $M$ and $N$ have MGU $P$}\\[0.75em]
\infer{\Delta \gives 0 \unif 0 \gen (\Delta \mid \id)}{} \quad
\infer{\Delta \gives \suc M \unif \suc N \gen P}{\Delta \gives M \unif N \gen P} \quad
\infer{\Delta \gives u \unif u \gen (\Delta \mid \id)}{}
\\[1em]
\infer{\Delta \gives u \unif M \gen (\Delta' \mid \id_1,M/u,\id_2)}
      {u \notin \FV(M) & \Delta = \Delta_1, u \ann U, \Delta_2 & \Delta' = \Delta_1, \Delta_2[M/u]} \quad
  (\text{same for } M \unif u)
\\[1em]
\infer{\Delta \gives 0 \unif \suc M \gen \fail}{} \quad
\infer{\Delta \gives \suc M \unif 0 \gen \fail}{} \quad
\infer{\Delta \gives u \unif M \gen \fail}{u \in \FV(M) & M \neq u} \quad
  (\text{same for } M \unif u)
\end{array}
\]
\LONGVERSION{The unification procedure is required for type checking the equality elimination forms
$\eqelim{s}{\Theta}{\Delta\pr}{t}$ and $\eqelimfalse{s}$ in \tores, which we explain in
\cref{sec:tores-terms}.
In each form, the term $s$ is a witness of an index equality $\Delta \gives M = N$.
In order to use this equality (or determine that it is spurious), we perform unification
$\Delta \gives M \unif N \gen P$ and check that the result $P$ matches the source term.
For the term $\eqelim{s}{\Theta}{\Delta\pr}{t}$, we check that $P$ is an $\alpha$-equivalent unifier
to the provided one $(\Delta\pr \mid \Theta)$.
For the failure term $\eqelimfalse{s}$ we check that $P$ is $\fail$, yielding a contradiction.
Hence our type checking algorithm for \tores\/ relies on a sound and complete unification procedure.
We summarize our requirements for unification below.
}{}
\begin{req}[Decidable unification]
\label{req:dec-unif}
  Given index terms $M$ and $N$ in a context $\Delta$, the judgment $\Delta \gives M \unif N \gen P$
  is decidable.
  Either $P$ is $(\Delta\pr \mid \Theta)$, the unique MGU up to $\alpha$-equivalence, or $P$ is
  $\fail$ and there is no unifier.
\end{req}

Finally, our operational semantics relies on index \emph{matching}.
This is an asymmetric form of unification: given terms $M$ in $\Delta$ and $N$ in $\Delta\pr$,
matching identifies a substitution $\Theta$ such that $\Delta\pr \gives M[\Theta] = N$.
We describe it using the judgment $\Delta \gives M \match N \synth (\Delta\pr \mid \Theta)$.
\LONGVERSION{

Matching is used during evaluation of the equality elimination $\eqelim{s}{\Theta}{\Delta\pr}{t}$ to
extend the substitution $\Theta$ to a full index environment (grounding substitution) $\theta$.
To achieve this, we must lift the notion of matching to the level of index substitutions.
This can be done generically given an algorithm for matching index terms.
The judgment $\Delta \gives \Theta_1 \match \Theta_2 \synth (\Delta\pr \mid \Theta)$ says that
matching discovered a substitution $\Theta$ such that $\Delta\pr \gives \Theta_1[\Theta] = \Theta_2$.

To illustrate an algorithm for index matching, we provide the rules for our natural number domain in
\cref{fig:matching}.
We also show the generic lifting of the algorithm to match index substitutions.

\begin{figureone}{
  Index matching for natural numbers and generic substitutions
  \label{fig:matching}
}
\[
\begin{array}{c}
\multicolumn{1}{l}{
  \fbox{$\Delta \gives M \match N \synth (\Delta\pr \mid \Theta)$}
  \quad
  \mbox{Index term $M$ matches index term $N$ under $\Theta$}
}\\[0.75em]
\infer{\Delta \gives u \match N \synth (\Delta\pr \mid \id_{\Delta_1},\sub{N}{u},\id_{\Delta_2})}
      {\Delta = \Delta_1, u:\nat, \Delta_2 & \Delta\pr = \Delta_1, \Delta_2}
\\[0.75em]
\infer{\Delta \gives z \match z \synth (\Delta \mid \id_{\Delta})}{}
\quad
\infer{\Delta \gives \suc M \match \suc N \synth (\Delta\pr \mid \Theta)}
      {\Delta \gives M \match N \synth (\Delta\pr \mid \Theta)}
\\[0.75em]
\multicolumn{1}{l}{
  \fbox{$\Delta \gives \Theta_1 \match \Theta_2 \synth (\Delta\pr \mid \rho)$}\quad
  \parbox{8.5cm}{Index substitution $\Theta_1$ matches index substitution $\Theta_2$ under $\rho$}
}\\[0.75em]
\infer{\Delta \gives \empt \match \empt \synth (\Delta \mid \id_{\Delta})}{}
\quad
\infer{\Delta \gives (\Theta_1,M/u) \match (\Theta_2, N/u) \synth (\Delta_2 \mid \rho_1[\rho_2])}
      {\Delta \gives \Theta_1 \match \Theta_2 \synth (\Delta_1 \mid \rho_1) &
       \Delta_1 \gives M[\rho_1] \match N \synth (\Delta_2 \mid \rho_2)}
\end{array}
\]
\end{figureone}

We then require that index (substitution) matching is both sound and complete.
We make these properties precise in our final requirements below.
}{The notion of matching also lifts to the level of index substitutions.
We omit the full specifications here and instead state the required properties.
}
\begin{req}[Soundness of index matching]$\;$\\[-1.5em]
  \counterwithin{enumi}{req}
  \begin{enumerate}
    \item
      If $\Delta \gives M \oft U$ and $\Delta \gives M \match N \synth (\Delta\pr \sep \Theta)$ then
      $\Delta\pr \gives \Theta \oft \Delta$ and $\Delta\pr \gives M[\Theta] = N$.
    \item\label{thm:sound-match}
      If $\Delta_1 \gives \Theta_1 \oft \Delta$ and
      $\Delta_1 \gives \Theta_1 \match \Theta_2 \synth (\Delta_2 \sep \Theta)$
      then
      $\Delta_2 \gives \Theta \oft \Delta_1$ and
      $\Delta_2 \gives \Theta_1[\Theta] = \Theta_2$.
  \end{enumerate}
\end{req}

\begin{req}[Completeness of index matching]
\label{req:match-comp}
  Suppose $\gives \theta \oft \Delta$ and $\gives M[\theta] = N[\theta]$ and
  $\Delta \gives M \unif N \gen (\Delta\pr \mid \Theta)$.
  Then $\Delta\pr \gives \Theta \match \theta \synth (\empt \sep \theta\pr)$.
\end{req}

\section{Specification of \tores}
\label{sec:tores-spec}
We now describe \tores\/, a programming language designed to express
(co)inductive proofs and programs using Mendler-style (co)recursion.
It also features \emph{index-stratified} types, which allow
definitions of types via well-founded recursion over indices.

\subsection{Types and Kinds}
\label{sec:tores-types}

Besides unit, products and sums, \tores\/ includes a nonstandard function type
$\iarr{\many{u \ann U}}{T_1}{T_2}$, which combines a dependent function type and a simple function
type.
It binds a number of index variables $\many{u \ann U}$ which may appear in both $T_1$ and $T_2$.
If the spine of type declarations is empty then $\iarr {\cdot}{T_1}{T_2}$ degenerates to the simple
function space.
We can also quantify existentially over an index using the type $\sigt{u \ann U}{T}$, and have a
type for index equality $M = N$.
These two types are useful for expressing equality constraints on
indices.
We model (co)recursive and stratified types as type constructors of kind $\pik{\many{u \ann U}} \type$.
These introduce type variables $X$, which we track in the type variable context $\Xi$.
There is no positivity condition on recursive types, as the typing rules for Mendler-recursion
enforce termination without it.

A stratified type is defined by primitive recursion on an index term.
For the index type $\nat$, the two branches correspond to the two constructors $0$ and $\suc$.
Intuitively, $T_{\Rec} \App 0$ will behave like $T_0$ and $T_{\Rec} \App (\suc M)$ will behave like
$T_s[\sub{M}{u},\sub{T_{\Rec} \App M}{X}]$.
For richer index languages such as Contextual LF we can generate an appropriate recursion scheme
following \citet{Pientka:TLCA15}.
\[
\begin{array}{ll@{~}r@{~}l}
\mbox{Kinds} & K &\bnfas& \type \sep \pik{u \ann U} K \\
\mbox{Types} & T &\bnfas& \Unit \sep T_1 \cross T_2 \sep T_1 \tsum T_2 \sep \iarr{\many{u \ann U}}{T_1}{T_2}
                      \sep \sigt{u \ann U}{T} \sep M = N \\
            && \sep & T \App M\sep \Lam u T \sep X \sep \rect{X \ann K} T \sep \corect{X \ann K} T \sep T_{\Rec} \\
\mbox{Stratified Types} & T_{\Rec} & \bnfas & \Rec_K (0 \mapsto T_0 \sep \suc u,\,X \mapsto T_s) \\
\mbox{Index Contexts}~& \Delta & \bnfas & \cdot \sep \Delta, u \ann U \\
\mbox{Type Var. Contexts} & \Xi &\bnfas& \empt \sep \Xi, X \ann K \\
\mbox{Typing Contexts}~& \Gamma & \bnfas & \cdot \sep \Gamma, x \ann T
\end{array}
\]
\begin{example}\label{ex:vec-types}
  We illustrate indexed recursive types and stratified types using vectors, i.e. lists indexed by their
  length, with elements of type $A$.
  Vectors are of kind $\Pi n{:}\nat.\type$.
  We omit the kind annotation for better readability in the subsequent type definitions.
  One way to define vectors is with an indexed recursive type, an explicit equality and an existential type:
  $ \Vect_{\mu} \equiv \rect{V} \Lam{n}    n = 0 \tsum \sigt{m \ann \nat} n = \suc m \cross (A \cross V \App m)$.

  Alternatively, they can be defined as a stratified type: $\Vect_{S} \equiv \Recnat{\Unit}{m}{V}{A \cross V}{}$.
  In this case equality reasoning is implicit.
  While we have a choice how to define vectors, some types are only possible to encode using one form or the other.
\end{example}

\begin{example}
  A type that must be stratified is the encoding of reducibility for
  simply typed lambda terms.
  This example is explored in detail by \citet{Cave:LFMTP13};
  our work gives it theoretical justification.

  Here the index objects are the simple types, \texttt{unit} and $\mathtt{arr}~A~B$ of
  index type \texttt{tp}, as well as lambda terms (), $\mathtt{lam}~x.M$ and
  $\mathtt{app}~M~N$ of index type \texttt{tm}.
  We can define reducibility as a stratified type of kind
  $\Pi a{:}\mathtt{tp}.\Pi m{:}\mathtt{tm}.\type$.
  This relies on an indexed recursive type $\mathtt{Halt}$ (omitted here) that
  describes when a term $m$ steps to a value.
  \[
    \begin{array}{r@{~}l@{~}c@{~}l}
      \Red \equiv \Rec~( & \mathtt{unit} & \mapsto & \Lambda m. \mathtt{Halt}~m  \\
      \sep   & \mathtt{arr}~a~b, R_a, R_b & \mapsto & \Lambda m. \mathtt{Halt}~m \times \iarr{n\ann\mathtt{tm }}{ R_a~n}{R_b~(\mathtt{app}~m~n)} ~)
    \end{array}
  \]
\end{example}

\begin{example}
  To illustrate a corecursive type, we define an indexed stream of bits following
  \citet{Thibodeau:ICFP16}.
  The index here guarantees that we are reading exactly $m$ bits.
  Once $m = 0$, we read a new message consisting of the length of the message $n$ together with a stream indexed by $n$.
  In contrast to the recursive type definition for vectors, here the equality constraints guard the observations we can make about a stream. 
  \[
    \begin{array}{r@{\;}r@{~}l@{~}l}
      \Stream \equiv \corect{\text{Str}} \Lam m & \iarr{\empt}{m =& 0 &}{\sit{n \ann \nat}{\text{Str}\;n}} \\
        \cross & \iarr{n \ann \nat}{m =& \suc n &}{\text{Str}\;n} \\
        \cross & \iarr{n \ann \nat}{m =& \suc n &}{\text{Bit}} \\
    \end{array}
  \] 
\end{example}

\subsection{Terms}
\label{sec:tores-terms}

\tores\/ contains many common constructs found in functional programming languages, such as unit,
pairs and case expressions.
We focus on the less standard constructs: indexed functions, equality witnesses, well-founded
recursion and index induction.

\[
\begin{array}{l@{~}r@{~}l}
  \mbox{Terms}~t,s & \bnfas & x \sep \unit \sep \ilam{\vec u}{x} t \sep t \app \vec M \app s \sep \pair{t_1}{t_2} \sep \psplit{s}{x_1}{x_2} t \\
  &\sep& \inj{i} t \sep (\case t \inj{1}\,x_1 \mapsto t_1 \sep \inj{2}\,x_2 \mapsto t_2) \\
  &\sep & \pack{M}{t} \sep \unpack{t}{u}{x} s \\
  &\sep & \refl \sep \eqelim{s}{\Theta}{\Delta}{t} \sep \eqelimfalse s\\
  &\sep & \inj{\mu} t \sep \rec f t \sep \corec f t \sep \out{\nu} t \sep \inj{l} t \sep \out{l} t \sep \indnat{t_0}{u}{f}{t_s} \sep t \ann T
\end{array}
\]

Since we combine the dependent and simple function types in $\iarr{\many{u \ann U}}{T_1}{T_2}$, we
similarly combine abstraction over index variables $\vec u$ and a term variable $x$ in our function
term $\ilam{\vec u}{x} t$.
The corresponding application form is $t \app \vec M \app s$.
The term $t$ of function type $\iarr{\many{u \ann U}}{T_1}{T_2}$ receives first a spine $\vec M$ of index objects followed by a term $s$.
Each equality type $M = N$ has at most one inhabitant $\refl$ witnessing the equality.
There are two elimination forms for equality: the term $\eqelim{s}{\Theta}{\Delta}{t}$ uses an
equality proof $s$ for $M = N$ together with a unifier $\Theta$ to refine the body $t$ in a new
index context $\Delta$.
It may also be the case that the equality witness $s$ is false, in which case we have reached a
contradiction and abort using the term $\eqelimfalse s$.
Both forms are necessary to make use of equality constraints that arise from indexed type
definitions and to show that some cases are impossible.

Recursive types are introduced by the ``fold'' syntax $\inj{\mu}$,
and stratified types are introduced by $\inj{l}$ terms.
Here $l$ ranges over constructors in the index language such as 
$0$ and $\suc$.
The important difference is how we eliminate recursive and stratified types.
We can analyze data defined by a recursive type using Mendler-style recursion $\rec f t$.
This gives a powerful means of recursion while still ensuring termination.
Stratified types can only be unfolded using  $\out{l}$ according to the index.
To take full advantage of stratified types, we also allow programmers to use well-founded recursion
over index objects, writing $\indnat{t_0}{u}{f}{t_s}$. Intuitively, if
the index object is $0$, then we pick  the first branch and execute
$t_0$; if the index object is $\suc M$ then we pick the second branch
instantiating $u$ with $M$ and allowing recursive calls $f$ inside
$t_s$. While this induction principle is specific to natural numbers,
it can also be derived for other index domains, in particular contextual LF (see \citet{Pientka:TLCA15}).

\begin{example}
\label{ex:copy-rec}
  Recall that vectors can be defined using the indexed recursive type
  $\Vect_\mu$ or the stratified type $\Vect_S$. 
Which definition we choose impacts how we write programs that analyze
vectors. We show the difference using a recursive function that copies a vector.
%
  \[
    \begin{array}{lcl}
      \multicolumn{3}{l}{copy\oft \iarr{n \ann \nat}{\Vect_\mu n}{\Vect_\mu n} \equiv \rec f \ilam{n}{v} \case v }\\
      \qquad \sep \inj{1} z & \mapsto & \inj{\mu} (\inj{1} z) \\
      \qquad \sep \inj{2} s & \mapsto & \unpack{s}{m}{p} \\
      \qquad &  & \psplit{p}{e}{p\pr} \\
      \qquad &  & \psplit{p\pr}{h}{t} \\
      \qquad &  &\inj{\mu} (\inj{2} (\pack{m}{\pair{e}{\pair{h}{f \app m \app t}}}))
    \end{array}
  \]
To analyze the recursively defined vector, we use recursion and case
analysis of the input vector to reconstruct the output vector. If we receive a non-empty list, we take it apart and expose the equality proofs, before
  reassembling the list.
  The recursion is valid according to the Mendler typing rule since the recursive call to
  $f$ is made on the tail of the input vector.
  \LONGVERSION{The program is fairly verbose due to the need to unpack the $\Sigma$-type and to split pairs.
  We also need to inject values into the recursive type using the $\inj{\mu}$ tag. In general, we may also need to reason explicitly with equality constraints.}

 To contrast we show the program using induction on natural numbers and unfolding the stratified type
  definition of $\Vect_S$.
  Note that the first argument is the natural number index $n$ paired with a unit term argument,
  since index abstraction is always combined with term abstraction.
  The program analyzes $n$ and in the $\suc$ case unfolds the input vector before
  reconstructing it using the result of the recursive call.
  In this version of $copy$ the equality constraints are handled silently by the type checker.
\[
  \begin{array}{r@{~}lcl}
    \multicolumn{4}{l}{ copy \oft \iarr{n \ann \nat}{\Unit}{\iarr{\empt}{\Vect_S n}{\Vect_S n}} \equiv} \\
\ind (            &                0 & \mapsto & \lam v \inj{0} \unit \\
\qquad \sep &\suc m,\, f_m & \mapsto & \lam v \psplit{(\out{\suc} v)}{h}{t} ~\;\inj{\suc} \pair{h}{f_m \app t})
  \end{array}
\]
\end{example}
\LONGVERSION{
\begin{example}
  Let us now build streams using Mendler-style corecursion.  Streams
  of natural numbers are defined as $\corect{X \ann K} \nat \times
  X$. We can define a stream of natural numbers starting from $n$ as:
  \[
   natsFrom \oft \nat \to \Stream \equiv \corec f {\lam n \pair{n}{f\ (\suc\ n)}}
  \]
  Here, the corecursion constructor takes a function whose argument is
  the seed used to build the stream and produces the pair of the head
  and the tail of our stream. In this case, the seed is simply the
  natural number corresponding to the current head of the stream.
  As we move to the tail, we simply use the corecursive call made available
  by $f$ on the successor of the seed, as the next element of the stream
  will be this new number.

  Let us now a second example: the stream of Fibonacci
  numbers.
  \[
  fibFrom \oft \nat \times \nat \to \Stream \equiv
  \corec f {\lam s \psplit{s}{m}{n} \pair{m}{f\ \pair{n}{m + n}}}
  \]
  In order to build the Fibonacci stream, we use as a seed a pair of
  two numbers corresponding of the last two Fibonacci numbers that
  have been computed so far. From there, the head of the stream is
  simply the first number of the pair, while the new seed is simply
  the second number together with the sum of the two, representing the
  second and third number, respectively.  Hence, to obtain the whole
  Fibonacci stream, we simply write $fibFrom\ \pair 0 1$.
\end{example}}{}

\begin{example}
  Note that \tores\/ does not have an explicit notion of falsehood.
  This is because it is definable using existing constructs: we can define the empty type as
  a recursive type $\bot \equiv \rect{X \ann \type} X$, and a contradiction term
  $\abort \equiv \rec f f \oft \bot \to C$, for any type $C$.
  Our termination result with the logical relation in \cref{sec:typ-interp} shows
  that the $\bot$ type contains no values and hence no closed terms, which implies logical consistency
  of \tores\/ (not all propositions can be proven).
\end{example}

\subsection{Typing Rules}
\label{sec:tores-typing}

We define a bidirectional type system in \cref{fig:tores-typing}\LONGVERSION{~with two mutually defined judgments:
checking a term $t$ against a type $T$ and synthesizing a type $T$ for a term $t$.}{.}
\LONGVERSION{We can move from checking to synthesis via the conversion
  rule and from synthesis to checking using a type annotation. The typing rules for unit, products, sums and
  existentials are standard.}{We focus here on equality, recursive and stratified types.}

\begin{figureone}{
  Typing rules for \tores
  \label{fig:tores-typing}
}
\begin{small}
\[
\begin{array}{l}
\fbox{$\Delta;\Xi;\Gamma \gives t \chk T$}\qquad\mbox{Term $t$ checks against input type $T$}
\\[0.75em]
\infer{\Delta;\Xi;\Gamma \gives \pair{t_1}{t_2} \chk T_1 \times T_2}
      {\Delta;\Xi;\Gamma \gives t_1 \chk T_1 & \Delta;\Xi;\Gamma \gives t_2 \chk T_2}
\quad
\infer{\Delta;\Xi;\Gamma \gives \psplit{p}{x_1}{x_2} t \chk T}
      {\Delta;\Xi;\Gamma \gives p \syn T_1 \times T_2 &
       \Delta;\Xi;\Gamma,x_1 \ann T_1,x_2 \ann T_2 \gives t \chk T}
\\[0.75em]
\infer{\Delta;\Xi;\Gamma \gives \inj{i} t \chk T_1 \tsum T_2}{\Delta;\Xi;\Gamma \gives t \chk T_i}
\quad
\infer{\Delta;\Xi;\Gamma \gives (\case t \inj{1} x_1 \mapsto t_1 \sep \inj{2} x_2 \mapsto t_2) \chk S}
      {\Delta;\Xi;\Gamma \gives t \syn T_1 \tsum T_2 &
       \Delta;\Xi;\Gamma,x_1 \ann T_1 \gives t_1 \chk S &
       \Delta;\Xi;\Gamma,x_2 \ann T_2 \gives t_2 \chk S}
\\[0.75em]
\infer{\Delta;\Xi;\Gamma \gives \pack{M}{t} \chk \sigt{u \ann U} T}
      {\Delta \gives M \oft U & \Delta;\Xi;\Gamma \gives t \chk T[\sub{M}{u}]}
\quad
\infer{\Delta;\Xi;\Gamma \gives \unpack{t}{u}{x} s \chk S}
      {\Delta;\Xi;\Gamma \gives t \syn \sigt{u \ann U}{T} & \Delta,u \ann U;\Xi;\Gamma,x \ann T \gives s \chk S}
\\[0.75em]
\infer{\Delta;\Xi;\Gamma \gives \refl \chk M = N}{\Delta \gives M = N}
\quad
\infer{\Delta;\Xi;\Gamma \gives \eqelimfalse s \chk T}
      {\Delta;\Xi;\Gamma \gives s \syn M = N &
       \Delta \gives M \unif N \gen \fail}
\\[0.75em]
\infer{\Delta;\Xi;\Gamma \gives \eqelim{s}{\Theta}{\Delta\pr}{t} \chk T}
      {\Delta;\Xi;\Gamma \gives s \syn M = N &
       \Delta \gives M \unif N \gen (\Delta\pr \mid \Theta) &
       \Delta\pr;\Xi[\Theta];\Gamma[\Theta] \gives t \chk T[\Theta]}
\\[0.75em]
\infer{\Delta;\Xi;\Gamma \gives \inj{\mu} t \chk (\rect{X \ann K} \Lam{\vec u} T) \App \vec M}
      {\Delta;\Xi;\Gamma \gives t \chk T[\many{\sub{M}{u}};\sub{\rect{X \ann K} \Lam{\vec u} T}{X}]}
\quad
\infer{\Delta;\Xi;\Gamma \gives \rec f t \chk \iarr{\many{u \ann U}}{(\rect{X \ann K} \Lam{\many{v}} S) \App \vec{u}}{T}}
      {\Delta;\Xi,X \ann K;\Gamma,f \ann (\iarr{\many{u \ann U}}{X \vec{u}}{T})
       \gives t \chk \iarr{\many{u \ann U}}{S[\many{\sub{u}{v}}]}{T}}
\\[0.75em]
\infer{\Delta;\Xi;\Gamma \gives \corec f t \chk \iarr{\many{u \ann U}}{S}{(\corect{X \ann K} \Lam{\many{v}} T) \App \vec{u}}}
      {\Delta;\Xi,X \ann K;\Gamma,f \ann (\iarr{\many{u \ann U}}{S}{X \vec{u}})
        \gives t \chk \iarr{\many{u \ann U}}{S}{T[\many{\sub{u}{v}}]}}
\quad
\infer{\Delta;\Xi;\Gamma \gives \ilam{\vec u}{x} t \chk \iarr{\many{u \ann U}}{S}{T}}
      {\Delta,\many{u \ann U};\Xi;\Gamma, x \ann S \gives t \chk T}
\\[0.75em]
\infer{\Delta;\Xi;\Gamma \gives \indnat{t_0}{u}{f}{t_s} \chk \iarr{u \ann \nat}{\Unit}{T}}
      {\Delta;\Xi;\Gamma \gives t_0 \chk T[\sub{0}{u}] &
        \Delta,u \ann \nat;\Xi;\Gamma,f \ann T \gives t_s \chk T[\sub{\suc u}{u}]}
\quad
\infer{\Delta;\Xi;\Gamma \gives \unit \chk 1}{}
\\[0.75em]
\infer{\Delta;\Xi;\Gamma \gives \inj{0} t \chk T_{\Rec} \App 0 \App \vec M}
      {\Delta;\Xi;\Gamma \gives t \chk T_0 \App \vec M
      }
\quad
\infer
      {\Delta;\Xi;\Gamma \gives \inj{\suc} t \chk T_{\Rec} \App (\suc N) \App \vec M}
      {\Delta;\Xi;\Gamma \gives t \chk T_s[\sub{N}{u};\sub{(T_{\Rec} \App N)}{X}] \App \vec M
      }
\quad
\infer{\Delta;\Xi;\Gamma \gives t \chk T}{\Delta;\Xi;\Gamma \gives t \syn T}
\\[0.75em]
\fbox{$\Delta;\Xi;\Gamma \gives t \syn T$}\qquad\mbox{Term $t$ synthesizes output type $T$}
\\[0.75em]
\infer{\Delta;\Xi;\Gamma \gives x \syn T}{x \ann T \in \Gamma}
\quad
\infer{\Delta;\Xi;\Gamma \gives t \app \vec M \app s \syn T[\many{\sub{M}{u}}]}
      {\Delta;\Xi;\Gamma \gives t \syn \iarr{\many{u \ann U}}{S}{T} &
       \Delta \gives \vec{M} \oft (\many{u\ann U}) &
       \Delta;\Xi;\Gamma \gives s \chk S[\many{\sub{M}{u}}]}
\\[0.75em]
\infer{\Delta;\Xi;\Gamma \gives \out{0} t \syn T_0 \App \vec M}
      {\Delta;\Xi;\Gamma \gives t \syn T_{\Rec} \App 0 \App \vec M}
\quad
\infer{\Delta;\Xi;\Gamma \gives \out{\suc} t \syn T_s[\sub{N}{u};\sub{(T_{\Rec} \App N)}{X}] \App \vec M}
      {\Delta;\Xi;\Gamma \gives t \syn T_{\Rec} \App (\suc N) \App \vec M}
\qquad
\infer{\Delta;\Xi;\Gamma \gives t \ann T \syn T}{\Delta;\Xi;\Gamma \gives t \chk T}
\\[0.75em]
\infer{\Delta;\Xi;\Gamma \gives \out{\nu} t \syn T[\many{\sub{M}{u}};\sub{\corect{X \ann K} \Lam{\vec u} T}{X}]}
      {\Delta;\Xi;\Gamma \gives t \syn (\corect{X \ann K} \Lam{\vec u} T) \App \vec M}
\end{array}
\]
\end{small}
\end{figureone}

The introduction for an index equality type is simply $\refl$, which is checked via equality
in the index domain.
Both equality elimination forms rely on unification in the index domain (see \cref{sec:index-unif}).
Specifically, the $\eqelimfalse s$ term checks against any type because the unification must fail,
establishing a contradiction.
For the term $\eqelim{s}{\Theta}{\Delta\pr}{t}$, unification must result in the MGU which by
\cref{req:dec-unif} is $\alpha$-equivalent to the supplied unifier $(\Delta\pr \mid \Theta)$.
We then check the body $t$ using the new index context $\Delta\pr$ and $\Theta$ applied to the
contexts $\Xi$ and $\Gamma$ and the goal type $T$.

This treatment of equality elimination is similar to the use of refinement substitutions for
dependent pattern matching \citep{Pientka:PPDP08,Cave:POPL12}, and is inspired by equality
elimination in proof theory \citep{Tiu:JAL12,McDowell:TOCL02,SchroederHeister93}.
In the latter line of work, type checking involves trying all unifiers from a \emph{complete set of
unifiers} (which may be infinite!), instead of a single most general unifier.
We believe our requirement for a unique MGU is a practical choice for type checking.

Indexed recursive and stratified types are both introduced by injections ($\inj{\mu}$ and
$\inj{l}$), though their elimination forms are different.
Stratified types are eliminated (unfolded) in reverse to the corresponding fold rules.
For recursive types on the other hand, the naive unfold rules lead to nontermination, so we use a
Mendler-style recursion form $\rec f t$, generalizing the original formulation \citep{Mendler:PHD}
to an indexed type system.
The idea is to constrain the type of the function variable $f$ so that it can only be applied to
structurally smaller data. This is achieved by declaring $f$ of type
 $\iarr{\many{u\ann U}}{X~\vec u} T$ in the premise of the rule.
Here $X$ represents types exactly one constructor smaller than the recursive type,
so the use of $f$ is guaranteed to be well-founded.

\begin{theorem}
\label{thm:dec-typ-chk}
  Type checking of terms is decidable.
\end{theorem}
\begin{proof}
  Since the typing rules are syntax directed, it is straight-forward to extract a type checking
  algorithm.
  Note that the algorithm relies on decidability of judgments in the index language, namely
  index type checking (\cref{req:dec-idx-typ-chk}), equality (\cref{req:dec-idx-eq}) and
  unification (\cref{req:dec-unif}).
\end{proof}

\newcommand{\close}[3]{(#1)[#2;#3]}
\newcommand{\msub}[2]{#2[#1]}
\newcommand{\mor}{~\text{or}~}

\subsection{Operational Semantics}
\label{sec:tores-eval}

We define a big-step operational semantics using environments, which provide closed values for the
free variables that may occur in a term.
%
\[
\begin{array}{lll}
  \text{Term environments} & \sigma &:= \empt \sep \sigma, v/x \\
  \text{Function values} & g &:= \ilam{\vec u}{x} t \sep \rec f t  \sep \corec f t \sep \indnat{t_0}{u}{f}{t_s} \\
  \text{Closures} & c &:= (g)[\theta;\sigma] \sep (\corec f t)[\theta;\sigma] \cdot \vec N\ v \\
  \text{Values} & v &:= c \sep \unit \sep \pair{v_1}{v_2} \sep \inj{i} v \sep \pack{M}{v} \sep \refl \sep \inj{\mu} v \sep \inj{l} v \\
\end{array}
\]
Values consist of unit, pairs, injections, reflexivity, and closures.
Typing for values and environments, which is used to state the subject
reduction theorem, are given \LONGVERSION{in \cref{fig:valtyping} in the appendix}{in the appendix}.

The main evaluation judgment, $t[\theta;\sigma] \evto v$,
describes the evaluation of a term $t$ under environments $\theta;\sigma$ to
a value $v$. Here, $t$ stands for a term in an index context $\Delta$
and term variable context $\Gamma$. The index environment $\theta$
provides closed index objects for all the index variables in $\Delta$, while
$\sigma$ provides closed values for all the variables declared in
$\Gamma$, i.e. $\gives \theta \oft \Delta$ and $\sigma \oft \Gamma[\theta]$.
For convenience, we factor out the application of a
closure $c$ to values $\vec N$ and $v$ resulting in a value $w$, using
a second judgment written $\appval{c}{\vec N}{v} \evto w$.
This allows us to treat application of functions (lambdas, recursion and
induction) uniformly. Similarly, we factor out the application of $\out{\nu}$
to a closure $c$ in an additional judgment written $\appout{c} \evto w$. This
simplifies the type interpretation used to prove termination.

\begin{figureone}{
  Big-step evaluation rules
  \label{fig:tores-eval}
}
\begin{small}
\[
\begin{array}{l}
\fbox{$\eval{t}{\theta;\sigma}{v}$}
\quad
\mbox{Term $t$ under environments $\theta$ and $\sigma$ evaluates to $v$}
\\[0.8em]
\infer{\eval{x}{\theta;\sigma}{v}}{\sigma(x) = v}
\quad
\infer{\eval{\unit}{\theta;\sigma}{\unit}}{}
\quad
\infer{\eval{\pair{t_1}{t_2}}{\theta;\sigma}{\pair{v_1}{v_2}}}{\eval{t_1}{\theta;\sigma}{v_1} & \eval{t_2}{\theta;\sigma}{v_2}}
\quad
\infer{\eval{(\psplit{t}{x_1}{x_2} s)}{\theta;\sigma}{v}}
      {\eval{t}{\theta;\sigma}{\pair{v_1}{v_2}} & \eval{s}{\theta;\sigma,\sub{v_1}{x_1},\sub{v_2}{x_2}}{v}}
\\[0.75em]
\infer{\eval{(\inj{i} t)}{\theta;\sigma}{\inj{i} v}}{\eval{t}{\theta;\sigma}{v}}
\quad
\infer{\eval{(\casetwo{t}{x_1}{t_1}{x_2}{t_2})}{\theta;\sigma}{v}}
      {\eval{t}{\theta;\sigma}{\inj{i} v'} & \eval{t_i}{\theta;\sigma,\sub{v'}{x_i}}{v}}
\quad
\infer{\eval{(t \ann T)}{\theta;\sigma}{v}}{\eval{t}{\theta;\sigma}{v}}
\\[0.75em]
\infer{\eval{(\pack{M}{t})}{\theta;\sigma}{\pack{M[\theta]}{v}}}{\eval{t}{\theta;\sigma}{v}}
\quad
\infer{\eval{(\unpack{t}{u}{x} s)}{\theta;\sigma}{v}}{\eval{t}{\theta;\sigma}{\pack{N}{v'}} & \eval{s}{\theta,\sub{N}{u};\sigma,\sub{v'}{x}}{v}}
\\[0.75em]
\infer{\eval{\refl}{\theta;\sigma}{\refl}}{}
\quad
\infer{\eval{(\eqelim{s}{\Theta}{\Delta}{t})}{\theta;\sigma} v}
      {\eval{s}{\theta;\sigma}{\refl} &
       \Delta \gives \Theta \match \theta \synth (\empt \sep \theta\pr) &
       \eval{t}{\theta\pr;\sigma}{v}}
\quad
\infer{\eval{(\inj{l} t)}{\theta;\sigma}{\inj{l} v}}{\eval{t}{\theta;\sigma}{v}}
\\[0.75em]
\infer{\eval{(\ilam{\vec u}{x}{t})}{\theta;\sigma}{(\ilam{\vec u}{x}{t})[\theta;\sigma]}}{}
\quad
\infer{\eval{(\rec f t)}{\theta;\sigma}{(\rec f t)[\theta;\sigma]}}{}
\quad
\infer{\eval{(\out{l} t)}{\theta;\sigma}{v}}{\eval{t}{\theta;\sigma}{\inj{l} v}}
\\[0.75em]
\infer{\eval{(\corec f t)}{\theta;\sigma}{(\corec f t)[\theta;\sigma]}}{}
\quad
\infer{\eval{(\out{\nu} t)}{\theta;\sigma}{w}}
	  {\eval{t}{\theta;\sigma}{c} & \appout{c} \evto w}
\\[0.75em]
\infer{\eval{(\indnat{t_0}{u}{f}{t_s})}{\theta;\sigma}{(\indnat{t_0}{u}{f}{t_s})[\theta;\sigma]}}{}
\quad
\infer{\eval{(t \app \vec M \app s)}{\theta;\sigma}{w}}{\eval{t}{\theta;\sigma}{c} & \eval{s}{\theta;\sigma}{v} & \appval{c}{\many{M[\theta]}}{v} \evto w}
\\[0.75em]
\fbox{$\appval{c}{\vec N}{v} \evto w$}
\quad
\mbox{Closure $c$ applied to values $\vec N$ and $v$ evaluates to $w$}
\\[0.8em]
\infer{\evalapp{(\ilam{\vec u}{x} t)[\theta;\sigma]}{\vec N}{v}{w}}
      {t[\theta,\many{\sub{N}{u}};\sigma,\sub{v}{x}] \evto w}
\quad
\infer{\evalapp{(\rec f t)[\theta;\sigma]}{\vec N}{(\inj{\mu} v)}{w}}
      {t[\theta;\sigma,\sub{(\rec f t)[\theta;\sigma]}{f}] \evto c & \evalapp{c}{\vec N}{v}{w}}
\\[0.75em]
\infer{\evalapp{(\corec f t)[\theta;\sigma]}{\vec N}{v}
               {(\corec f t)[\theta;\sigma] \cdot \vec N\ v}}{}

\\[0.75em]
\infer{\evalapp{(\indnat{t_0}{u}{f}{t_s})[\theta;\sigma]}{0}{\unit}{w}}{\eval{t_0}{\theta;\sigma}{w}}
\quad
\infer{\evalapp{(\indnat{t_0}{u}{f}{t_s})[\theta;\sigma]}{(\suc N)}{\unit}{w}}
      {\evalapp{(\indnat{t_0}{u}{f}{t_s})[\theta;\sigma]}{N}{\unit}{v} &
        \eval{t_s}{\theta,\sub{N}{u};\sigma,\sub{v}{f}}{w}}
\\[0.75em]
\fbox{$\appout{c} \evto w$}
\quad
\mbox{Closure $c$ applied to observation $\out{\nu}$ evaluates to $w$}
\\[0.8em]
\infer{\appout{((\corec f t)[\theta;\sigma] \cdot \vec N\ v)} \evto w}
      {t[\theta;\sigma,\sub{(\corec f t)[\theta;\sigma]}{f}] \evto c
        & \evalapp{c}{\vec N}{v}{w}}
\end{array}
\]
\end{small}
\end{figureone}

We only explain the evaluation rule for equality elimination $\eqelim{s}{\Theta}{\Delta}{t}$.
We first evaluate the equality witness $s$ under environments $\theta;\sigma$ to the value $\refl$.
This ensures that $\theta$ respects the index equality $M = N$ witnessed by $s$.
From type checking we know that $\Delta \gives M[\Theta] = N[\Theta]$: the key is how we extend
$\Theta$ at run-time to produce a new index environment $\theta\pr$ that is consistent with $\theta$.
This relies on sound and complete index substitution matching (see \cref{sec:index-unif}) to
generate $\theta\pr$ such that $\empt \gives \theta\pr \oft \Delta$ and
$\empt \gives \Theta[\theta\pr] = \theta$.
We can then evaluate the body $t$ under the new index environment $\theta\pr$ and the same term
environment $\sigma$ to produce a value $v$.

Notably absent is an evaluation rule for $\eqelimfalse t$.
This term is used in a branch of a case split that we know statically to be impossible.
Such branches are never reached at run time, so there is no need for an evaluation rule.
For example, consider a type-safe ``head'' function, which receives a nonempty vector as input.
As we write each branch of a case split explicitly, the empty list case must use $\eqelimfalse t$,
but is never executed.
We now state subject reduction for \tores.

\begin{restatable}[Subject Reduction]{theorem}{subred}$\;$\\[-1.5em]
\label{thm:sub-red}
  \begin{enumerate}
  \item
    If $\eval t {\theta;\sigma} v$ where
    $\Delta;\cdot ;\Gamma \gives t \chk T$ or $\Delta;\cdot;\Gamma \gives t \syn T$,
    and $\gives \theta : \Delta$ and $\sigma : \Gamma[\theta]$,
    then $v : T[\theta]$.
  \item
    If $\evalapp{g[\theta;\sigma]}{\vec N}{v}{w}$ where
    $\Delta;\empt;\Gamma \gives g \chk \iarr{\many{u\ann U}}{S}{T}$ and
    $\gives \theta : \Delta$ and $\sigma : \Gamma[\theta]$ and
    $\gives \vec N \oft (\many{u\ann U})[\theta]$ and
    $v : S[\theta,\many{N/u}]$,
    then $w: T[\theta,\many{N/u}]$.
  \item
    If $\appout{c} \evto w$ where $c : (\corect {X{:}K} \Lam {\vec u} T)\ \vec M$
    then $w : T[\many{\sub{M}{u}}; (\corect {X{:}K} \Lam {\vec u} T)/X]$.

  \end{enumerate}
\end{restatable}

\section{Termination Proof}
\label{sec:tores-proof}
We now describe our main technical result: termination of evaluation.
Our proof uses the logical predicate technique of \cite{Tait67} and \cite{Girard1972}.
We interpret each language construct (index types, kinds, types, etc.) into a semantic model of sets
and functions.

\subsection{Interpretation of Index Language}

We start with the interpretations for index types and spines.
In general, our index language may be dependently typed, as it is if we choose Contextual LF.
Hence our interpretation for index types $U$ must take into account an environment $\theta$
containing instantiations for index variables $u$.
Such an index environment $\theta$ is simply a grounding substitution
$\gives \theta \oft \Delta$.

\begin{definition}[Interpretation of index types $\sem U$ and index
  spines $\sem {\many{u \ann U}}$]\label{def:sem_spine}
\[
  \begin{array}{lcl}
   \sem{U}(\theta) & = & \{ M \sep \empt \gives M \oft U[\theta] \}\\[0.5em]
    \sem{(\empt)}(\theta) &=& \{ \empt \} \\
    \sem{(u_0 \ann U_0, \many{u \ann U})}(\theta)
    &=& \{ M_0,\vec M \sep M_0 \in \sem{U_0}(\theta), \vec M \in \sem{(\many{u \ann U})}(\theta,\sub{M_0}{u_0}) \}
  \end{array}
\]
\end{definition}

The interpretation of an index type $U$ under environment $\theta$ is the set of closed terms of
type $U[\theta]$.
The interpretation lifts to index spines $(\many{u \ann U})$.
%
%
With these definitions, the following lemma follows from the substitution principles of index terms
(\cref{req:isubst}).

\begin{lemma}[Interpretation of index substitution]$\;$\\[-1.5em]
  \counterwithin{enumi}{theorem}
  \begin{enumerate}
    \item\label{lem:sem_isubst}
      If $\Delta \gives M \oft U$ and $\gives \theta \oft \Delta$ then $M[\theta] \in \sem{U}(\theta)$.
    \item\label{lem:sem_isubst_spine}
      If $\Delta \gives \vec M \oft (\many{u \ann U})$ and $\gives \theta \oft \Delta$ then
      $\many{M[\theta]} \in \sem{(\many{u \ann U})}(\theta)$.
  \end{enumerate}
\end{lemma}

\subsection{Lattice Interpretation of Kinds}

We now describe the lattice structure that underlies the interpretation of kinds in our
language.
The idea is that types are interpreted as sets of term-level values and type constructors as
functions taking indices to sets of values.
We call the set of all term-level values $\VAL$ and write its power set as $\powset \VAL$.
The interpretation is defined inductively on the structure of kinds.

\begin{definition}[Interpretation of kinds $\sem K$]
\[
  \begin{array}{lcl}
    \sem{\type}(\theta) &=& \powset \VAL \\
    \sem{\pik{u \ann U} K}(\theta) &=& \{ \C \sep \fall{M \in \sem{U}(\theta)} \C(M) \in \sem{K}(\theta,\sub{M}{u}) \}
  \end{array}
\]
\end{definition}

A key observation in our metatheory is that each $\sem{K}(\theta)$ forms a \emph{complete lattice}.
In the base case, $\sem{\type}(\theta) = \powset \VAL$ is a complete lattice under the subset
ordering, with meet and join given by intersection and union respectively.
For a kind $K = \pik{u \ann U} K\pr$, we induce a lattice structure on $\sem{K}(\theta)$ by lifting
the lattice operations pointwise.
Precisely, we define
\[ \A \leq_{\sem{K}(\theta)} \B \quad\text{iff}\quad
\fall{M \in \sem{U}(\theta)} \A(M) \leq_{\sem{K\pr}(\theta,\sub{M}{u})} \B(M). \]
The meet and join operations can similarly be lifted pointwise.

This structure is important because it allows us to define pre-fixed points for operators on the
lattice, which is central to our interpretation of recursive types.
Here we rely on the existence of arbitrary meets, as we take the meet over an impredicatively
defined subset of $\mc L$.

\begin{definition}[Mendler-style pre-fixed and post-fixed points]
\label{def:prefp}
Suppose $\mc L$ is a complete lattice and $\F : \mc L \to \mc L$.
Define $\semlfp_{\mc L} : (\mc L \to \mc L) \to \mc L$ by
\[
  \semlfp_{\mc L} \F = \meet \{ \C \in \mc L \sep \fall{\X \in \mc L} \X \leq_{\mc L} \C \implies \F(\X) \leq_{\mc L} \C \}.
\]
and $\semgfp_{\mc L} : (\mc L \to \mc L) \to \mc L$ by
\[
  \semgfp_{\mc L} \F = \join \{ \C \in \mc L \sep \fall{\X \in \mc L} \C \leq_{\mc L} \X \implies \C \leq_{\mc L} \F(\X) \}.
\]
\end{definition}
We will mostly omit the subscript denoting the underlying lattice $\mc L$ of the order
$\leq$ and pre-fixed and post-fixed points, $\semlfp$ and $\semgfp$.

Note that a usual treatment of recursive types would define the least pre-fixed point of a monotone
operator as $\meet \{ \C \in \mc L \sep \F(\C) \leq \C \}$ and the greatest post-fixed point
of a monotone operator as $\join \{ \C \in \mc L \sep \C \leq \F(\C) \}$, using the Knaster-Tarski theorem.
However, our unconventional definition (following \citet{Rao:LFMTP16}) more closely models
Mendler-style (co)recursion and does not require $\F$ to be monotone
(thereby avoiding a positivity restriction on recursive types).

\subsection{Interpretation of Types}
\label{sec:typ-interp}

In order to interpret the types of our language, it is helpful to define semantic versions of some
syntactic constructs.
We first define a semantic form of our indexed function type $\iarr{\many{u \ann U}}{T_1}{T_2}$,
which helps us formulate the interaction of function types with fixed points and recursion.

\begin{definition}[Semantic function space]
  \label{def:sem-funsp}
  For a spine interpretation $\vec \U$ and functions $\A, \B : \vec \U \to \powset \VAL$, define
$
    \semarr{\vec \U}{\A}{\B} = \{ c \sep \fall{\vec M \in \vec \U} \fall{v \in \A(\vec M)} \evalapp{c}{\vec M}{v}{w} \in \B(\vec M) \}.
$
\end{definition}

It will also be convenient to lift term-level $\inj{}$ tags to the level of sets and functions in
the lattice $\sem{K}(\theta)$.
We define the lifted tags $\injop{} : \sem{K}(\theta) \to \sem{K}(\theta)$ inductively on $K$.
If $\V \in \sem{\type}(\theta) = \powset \VAL$ then $\injop{} \V = \inj{} \V = \{ \inj{} v \sep v \in \V \}$.
If $\C \in \sem{\pik{u \ann U} K\pr}(\theta)$ then $(\injop{} \C)(M) = \injop{} (\C(M))$ for all
$M \in \sem{U}(\theta)$.
Essentially, the $\injop{}$ function attaches a tag to every element in the set produced after the
index arguments are received.

Dually we define $\outop{\nu} : \sem{K}(\theta) \to \sem{K}(\theta)$.
If $\V \in \sem{\type}(\theta) = \powset \VAL$ then $\outop{\nu} \V = \out{\nu} \V = \{ c \sep  \appout{c} \evto w \in \V \}$.
If $\C \in \sem{\pik{u \ann U} K\pr}(\theta)$ then $(\outop{\nu} \C)(M) = \outop{\nu} (\C(M))$ for all
$M \in \sem{U}(\theta)$.

Finally, we define the interpretation of type variable contexts $\Xi$.
These describe semantic environments $\eta$ mapping each type variable to an object in its
respective kind interpretation.
Such environments are necessary to interpret type expressions with free type variables.

\begin{definition}[Interpretation of type variable contexts $\sem \Xi$]
\label{def:sem_tvctx}
\[
  \begin{array}{lcl}
    \sem{\empt}(\theta) &=& \{ \empt \} \\
    \sem{\Xi,X \ann K}(\theta) &=& \{ \eta,\sub{\X}{X} \sep \eta \in \sem{\Xi}(\theta), \X \in \sem{K}(\theta) \} \\
  \end{array}
\]
\end{definition}

We are now able to define the interpretation of types $T$ under environments $\theta$ and $\eta$.
This is done inductively on the structure of $T$.

\begin{definition}[Interpretation of types and constructors]
\label{def:sem_typ}
  \[
  \begin{array}{lcl}
    \sem{\Unit}(\theta;\eta) &=& \{ \unit \} \\
    \sem{T_1 \cross T_2}(\theta;\eta) &=& \{ \pair{v_1}{v_2} \sep v_1 \in \sem{T_1}(\theta;\eta), v_2 \in \sem{T_2}(\theta;\eta) \} \\
    \sem{T_1 \tsum T_2}(\theta;\eta) &=& \inj{1} \sem{T_1}(\theta;\eta) \bigcup \inj{2} \sem{T_2}(\theta;\eta) \\
    \sem{\iarr{\many{u \ann U}}{T_1}{T_2}}(\theta;\eta) &=& \semarr{\sem{(\many{u \ann U})}(\theta)}{\mc T_1}{\mc T_2} \\
    &&\text{where } \mc T_i(\vec M) = \sem{T_i}(\theta,\many{\sub{M}{u}};\eta) \text{ for } i \in \{1, 2\} \\
    \sem{\sigt{u \ann U}{T}}(\theta;\eta) &=& \{ \pack{M}{v} \sep M \in \sem{U}(\theta), v \in \sem{T}(\theta,\sub{M}{u};\eta) \} \\
    \sem{T \App M}(\theta;\eta) &=& \sem{T}(\theta;\eta)(M[\theta]) \\
    \sem{M = N}(\theta;\eta) &=& \{ \refl \sep \gives M[\theta] = N[\theta] \} \\
    \sem{X}(\theta;\eta) &=& \eta(X) \\
    \sem{\Lam u T}(\theta;\eta) &=& (M \mapsto \sem{T}(\theta,\sub{M}{u};\eta)) \\
    \sem{\rect{X \ann K} T}(\theta;\eta) &=& \semlfp_{\sem{K}(\theta)}(\X \mapsto \injop{\mu}(\sem{T}(\theta;\eta,\sub{\X}{X}))) \\
    \sem{\corect{X \ann K} T}(\theta;\eta) &=& \semgfp_{\sem{K}(\theta)}(\X \mapsto \outop{\nu} (\sem{T}(\theta;\eta,\sub{\X}{X}))) \\
    \sem{\Recnat{T_z}{u}{X}{T_s}{K}}(\theta;\eta) &=& \REC_{\sem{K}(\theta)}\,(\injop{0} \sem{T_z}(\theta;\eta)) \\
    && \qquad(N \mapsto \X \mapsto \injop{\suc} \sem{T_s}(\theta,\sub{N}{u};\eta,\sub{\X}{X}))
  \end{array}
  \]
\\[-1.75em]
\mbox{~~where}
  \[
  \begin{array}{llllcl}
    \REC_{\mc L} : &\mc L \to &(\N \to \mc L \to \mc L) \to &\N &\to &\mc L \\
    \REC_{\mc L} \App &\C \App &\F \App &0 &= &\C \\
    \REC_{\mc L} \App &\C \App &\F \App &(\suc N) &= &\F \App N \App (\REC_{\mc L} \App \C \App \F \App N)
  \end{array}
  \]
\end{definition}

The interpretation of the indexed function type
$\sem{\iarr{\many{u \ann U}}{T_1}{T_2}}(\theta;\eta)$ contains closures which, when applied to values in
the appropriate input sets, evaluate to values in the appropriate output set.
The interpretation of the equality type $\sem{M = N}(\theta;\eta)$ is the set $\{\refl\}$ if
$\gives M[\theta] = N[\theta]$ and the empty set otherwise.
The interpretation of a recursive type is the pre-fixed point of the function obtained from the
underlying type expression.
Finally, interpretation of a stratified type built from $\Rec$ relies on an analogous semantic
operator $\REC$.
It is defined by primitive recursion on the index argument, returning the first argument in the base
case and calling itself recursively in the step case.
Note that the definition of $\REC$ is specific to the index type it recurses over.
We only use the index language of natural numbers here, so the appropriate set of index values is
$\sem \nat = \N$.

Last, we give the interpretation for typing contexts $\Gamma$, describing well-formed
term-level environments $\sigma$.

\begin{definition}[Interpretation of typing contexts]
\label{def:sem_typctx}
  \begin{align*}
    \sem{\empt}(\theta;\eta) &= \{ \empt \} \\
    \sem{\Gamma, x \ann T}(\theta;\eta) &= \{ \sigma, \sub{v}{x} \sep \sigma \in \sem{\Gamma}(\theta;\eta), v \in \sem{T}(\theta;\eta) \}
  \end{align*}
\end{definition}

\subsection{Proof}

We now sketch our proof using some key lemmas.
The following two lemmas concern the fixed point operators $\semlfp$ and $\semgfp$,
and are key for reasoning about (co)recursive types and Mendler-style (co)recursion.
These lemmas generalize those of \citet{Rao:LFMTP16} from the simply typed setting.

\begin{restatable}[Soundness of pre-fixed point]{lemma}{prefp}
\label{lem:prefp}
Suppose $\mc L$ is a complete lattice, $\F : \mc L \to \mc L$ and $\semlfp$ is as in Def.
\ref{def:prefp}. Then $\F(\semlfp \F) \leq \semlfp \F$.
\end{restatable}

\begin{restatable}[Function space from pre-fixed and post-fixed points]{lemma}{funprefp}
\label{lem:fun-prefp}
Let $\mc L = \vec \U \to \powset \VAL$ and $\B \in \mc L$ and $\F : \mc L \to \mc L$.
\begin{enumerate}
\item If $\fall{\X \in \mc L} c \in \semarr{\vec \U}{\X}{\B} \implies
  c \in \semarr{\vec{\U}}{\F \X}{\B}$,
  then $c \in \semarr{\vec \U}{\semlfp \F}{\B}$.
\item If $\fall{\X \in \mc L} c \in \semarr{\vec \U}{\B}{\X} \implies
  c \in \semarr{\vec{\U}}{\B}{\F \X}$,
  then $c \in \semarr{\vec \U}{\B}{\semgfp \F}$.
  \end{enumerate}
\end{restatable}

Another key result we rely on is that type-level substitutions associate with our semantic
interpretations.
Note that single index (and spine) substitutions on types are handled as special cases of the result for
simultaneous index substitutions.
We omit the definitions of type substitutions for brevity.

\begin{lemma}[Type-level substitution associates with interpretation]\label{lem:subst}$\;$\\
  Suppose $\Delta;\Xi \gives T \chk K$ or $\Delta;\Xi \gives T \syn K$,
  and $\gives \theta \oft \Delta\pr$ and $\eta \in \sem{\Xi\pr}(\theta)$.
  \begin{enumerate}
    \item
      If $\Delta\pr \gives \Theta \oft \Delta$ and $\Xi\pr = \Xi[\Theta]$ then
      $\sem{\Xi\pr}(\theta) = \sem{\Xi}(\Theta[\theta])$ and
      $\sem{T[\Theta]}(\theta;\eta) = \sem{T}(\Theta[\theta];\eta)$.
    \item
      If $\Delta = \Delta\pr$ and $\Xi = \Xi\pr,X \ann K$ and
      $\Delta\pr;\Xi\pr \gives S \chk K$ or $\Delta\pr;\Xi\pr \gives S \syn K$, then
      $\sem{T[\sub{S}{X}]}(\theta;\eta) = \sem{T}(\theta;\eta,\sub{\sem{S}(\theta;\eta)}{X})$.
  \end{enumerate}
\end{lemma}
\begin{proof}
  By induction on the structure of $T$.
\end{proof}

The next two lemmas concern recursive types and terms respectively.

\begin{restatable}[Recursive type contains unfolding]{lemma}{appprefp}
\label{lem:app_prefp}$\;$\\
  Let $R = \rect{X \ann K} \Lam{\vec u} S$ where
  $K = \pik{\many{u \ann U}} \type$ and $\Delta;\Xi \gives R \syn K$,
  and $\Delta \gives \vec M \oft (\many{u \ann U})$ and
  $\gives \theta \oft \Delta$ and $\eta \in \sem{\Xi}(\theta)$.
  Then $\inj{\mu} \sem{S[\many{\sub{M}{u}};\sub{R}{X}]}(\theta;\eta) \subseteq \sem{R \App \vec M}(\theta;\eta)$.
\end{restatable}

\begin{restatable}[Backward closure]{lemma}{backclos}
\label{lem:back_clos}$\;$\\
Let $t$ be a term, $\theta$ and $\sigma$ environments, and $\A, \B \in \vec \U \to \powset \VAL$.
\begin{enumerate}
  \item If $t[\theta;\sigma, \sub{(\rec f t)[\theta;\sigma]}{f}] \evto c\pr \in \semarr{\vec \U}{\A}{\B}$,
  then $(\rec f t)[\theta;\sigma] \in \semarr{\vec \U}{\injop{\mu} \A}{\B}$.
 \item If $t[\theta;\sigma, \sub{(\corec f t)[\theta;\sigma]}{f}] \evto c\pr \in \semarr{\vec \U}{\A}{\B}$,
  then $(\corec f t)[\theta;\sigma] \in \semarr{\vec \U}{\A}{\outop{\nu} \B}$.
\end{enumerate}
\end{restatable}

Our final lemma concerns the semantic equivalence of an applied stratified type with its unfolding.
Note that here we only state and prove the lemma for an index language of natural numbers.
For a different index language, one would need to reverify this lemma for the corresponding stratified type.
This should be straight-forward once the semantic $\REC$ operator is chosen to reflect the inductive
structure of the index language.

\begin{lemma}[Stratified types equivalent to unfolding]
\label{lem:trec}$\;$\\
  Let $T_{\Rec} \equiv \Recnat{T_z}{n}{X}{T_s}{K}$ where
  $K = \pik{n \ann \nat} \pik{\many{u \ann U}} \type$ and $\Delta;\Xi \gives T_{\Rec} \syn K$,
  and $\Delta \gives \vec M \oft (\many{u \ann U})$ and $\Delta \gives N \oft \nat$
  and $\gives \theta \oft \Delta$ and $\eta \in \sem{\Xi}(\theta)$. Then
  \begin{enumerate}
    \item
    $\sem{T_{\Rec} \App 0 \App \vec M}(\theta;\eta) = \inj{0} (\sem{T_z \App \vec M}(\theta;\eta))$ and
    \item
    $\sem{T_{\Rec} \App (\suc N) \App \vec M}(\theta;\eta) =
    \inj{\suc} (\sem{T_s[\sub{N}{n};\sub{(T_{\Rec} \App N)}{X}] \App \vec M}(\theta;\eta))$.
  \end{enumerate}
\end{lemma}

Finally we state the main termination theorem.

\begin{restatable}[Termination of evaluation]{theorem}{termination}
\label{thm:term}
  If $\Delta;\Xi;\Gamma \gives t \chk T$ or $\Delta;\Xi;\Gamma \gives t \syn T$, and
  $\gives \theta \oft \Delta$ and $\eta \in \sem{\Xi}(\theta)$ and $\sigma \in \sem{\Gamma}(\theta;\eta)$,
  then $t[\theta;\sigma] \evto v$ for some $v \in \sem{T}(\theta;\eta)$.
\end{restatable}

\section{Related Work}
Our work draws inspiration from two different areas: dependent type
theory and proof theory.

Dependent type theories often support large eliminations that are definitions of dependent types by primitive recursion.
\LONGVERSION{For example, a large elimination on a natural number is of the form $\Rec\;t\;T_0\;(X.\,T_{\suc})$,
similar to our stratified type.
In a dependent type theory, this large elimination reduces in the same way as term-level recursion,
depending on the value of the natural number expression $t$.}{They
reduce the same way as term-level recursion.}
\LONGVERSION{Large eliminations are important for increasing the expressive power of dependent type theories, in
particular allowing one to prove that constructors of inductive types are disjoint (e.g. $0 \neq 1$).
Jan Smith \citep{Smith89} gave an account of large eliminations (calling them
\emph{propositional functions}) as an extension of Martin-L\"of type theory.
\citet{Werner:1992} was able to prove strong normalization for a language (dependently typed System
F) with large eliminations over natural numbers.
Werner's proof needs to consider the normalizability of the natural number argument to the
large elimination.
Our interpretation of stratified types is simplified by the fact that the natural number argument
comes from an index language containing only normal forms.
}{}
In general, our work shows how to gain the power of large eliminations
in an indexed type system by simulating reduction on the level of types by unfolding stratified types in their typing rules.


In the world of proof theory, our core language corresponds to a first-order logic with equality,
inductive and stratified (recursive) definitions.
\citet{Momigliano:TYPES03,Tiu:JAL12} give comprehensive treatments of logics with induction and
co-induction as well as first-class equality.
They present their logics in a sequent calculus style and prove cut elimination \LONGVERSION{(i.e. that the cut
rule is admissible)}{} which implies consistency of the logics.
Their cut elimination proof extends Girard's proof technique of reducibility candidates, similar to
ours.
Note that they require strict positivity of inductive definitions, i.e. the head of a definition
(analogous to the recursive type variable) is not allowed to occur to the left of an implication.

\cite{Tiu:IJCAR12} also develops a first-order logic with stratified definitions similar to our
stratified types.
His notion of stratification comes from defining the ``level'' of a formula, which measures its size.
A recursive definition is then called stratified if the level does not increase from the head of the
definition to the body.
This is a more general formulation than our notion of stratification for types:
we require the type to be stratified exactly according to the structure of an index term, instead of
a more general decreasing measure.
However, we could potentially replicate such a measure by suitably extending our index language.

Another approach to supporting recursive definitions in proof theory is via a rewriting relation, as
in the Deduction Modulo system \citep{dowek:modulo}.
The idea of this system is to generalize a given first-order logic to account for a congruence
relation defined by a set of rewrite rules.
This rewriting could include recursive definitions in the same sense as Tiu.
\citet{Dowek:JSL03} show that such logics under congruences can be proven normalizing given general
conditions on the congruence.
\LONGVERSION{Specifically, they define the notion of a \emph{pre-model} of a congruence, whose existence is
sufficient to prove cut elimination of the logic.}{}
\citet{Baelde:LICS12} extend this work in the following way.
First, they present a first-order logic with inductive and co-inductive definitions, together with a
general form of equality.
They show strong normalization for this logic using a reducibility candidate argument.
Crucially, their proof is in terms of a pre-model which anticipates the addition of recursive
definitions via a rewrite relation.
Then they give conditions on the rewrite rules, essentially requiring that each definition follows a
well-founded order on its arguments.
Under these conditions, they are able to construct a pre-model for the relation, proving
normalization as a result.
Although their notion of stratification of recursive definitions is
slightly more general than ours, our treatment is perhaps more direct
as the rewriting of types takes place within our typing rules, and our
semantic model accounts for stratified types directly. 

From a programmer's view, the proof theoretic foundations give rise to
programs written using iterators; our use of Mendler-style
(co)recursion is arguably closer to standard programming practice. 
Mendler-style recursion schemes for term-indexed languages have been investigated by \cite{Ahn:thesis}.
Ahn describes an extension of System F$_{\omega}$ with erasable term indices, called F$_i$.
He combines this with fixed points of type operators, as in the Fix$_{\omega}$ language by
\cite{Abel:CSL04}, to produce the core language Fix$_i$.
In Fix$_i$, one can embed Mendler-style recursion over term-indexed data types by Church-style
encodings.

Fundamentally, our use of indices is more liberal than in Ahn's core languages.
In F$_i$, term indices are drawn from the same term language as programs.
They are treated polymorphically, in analogue with polymorphic type indices, i.e. they must have
closed types and cannot be analyzed at runtime.
Our approach is to separate the language of index terms from the language of programs.
In \tores\/, the indices that appear in types can be handled and analyzed at runtime, may be
dependently typed and have types with free variables.
This flexibility is crucial for writing inductive proofs over LF specifications as we do in Beluga.

\LONGVERSION{Another difference from Ahn's work is our treatment of Mendler-style recursion.
Ahn is able to embed a variety of Mendler-style recursion schemes via Church encodings, taking
advantage of polymorphism and type-level functions inherited from System F$_{\omega}$.
Our work does not include polymorphism and general type-level functions as we concentrate on a small
core language for inductive reasoning.
For this purpose, \tores\/ includes recursive types with a Mendler-style elimination form.
We believe this treatment is a more direct interpretation of Mendler recursion for indexed recursive types.
}

\section{Conclusion and Future Work}
We presented a core language \tores\/ extending an indexed type system with (co)recursive types and
stratified types.
We argued that \tores\/ provides a sound and powerful foundation for programming (co)inductive proofs, in
particular those involving logical relations.
This power comes from the (co)induction principles on recursive types given by Mendler-style (co)recursion
as well as the flexibility of recursive definitions given by stratified types.
Type checking in \tores\/ is decidable and types are preserved during evaluation.
The soundness of our language is guaranteed by our logical predicate semantics and termination
proof.
We believe that \tores\/ balances well the proof-theoretic power with a simple
metatheory (especially when compared with full dependent types).

An important question to investigate in the future is how to compile a
practical language that supports (co)pattern matching into the core language we
propose in \tores. \SHORTVERSION{Similarly it would also be interesting to explore
how our treatment of indexed recursive and stratified types could help
(or hinder) proof search.  }
Such issues are important to solve in order to create a productive user experience for dependently
typed programming and proving.

\LONGVERSION{
It would also be interesting to explore how our treatment of indexed recursive and stratified
types could help (or hinder) proof search.
Proof search is a fundamental technique to ease the development and maintenance of proofs, by
automatically generating parts of proof terms.
Like \citet{Baelde:LICS12}, we are curious to see how our treatment of recursive definitions can be
handled by search techniques, especially those derived from focusing \citep{Andreoli92}.
}


\LONGVERSION{\newpage}
\appendix
\section{Appendix}
\label{sec:appendix}

\subsection{Kinding}

\LONGSHORT{
Our kinding rules are shown in \cref{fig:kinding}.  We employ a
bidirectional kinding system to make it evident when kinds can be
inferred and when kinding annotations are necessary.  Kinding depends
on two contexts: the index context $\Delta$, since index variables may
appear in types and kinds, and the type variable context $\Xi$.  Note
that in general the kinds assigned to type variables in $\Xi$ may
depend on the index variables in $\Delta$.  The checking judgment
$\Delta;\Xi \gives T \chk K$ takes all expressions as inputs and
verifies that the type is well-kinded.  The inference judgment
$\Delta;\Xi \gives T \syn K$ takes the contexts and type as input to
produce a kind as output.

In our rules, the kind $\type$ of types is always checked, and may rely on inference via the
conversion rule.
In addition, type-level lambdas $\Lam u T$ are checked against a kind $\pik{u \ann U} K$ by
checking the body $T$ under an extended index context $\Delta, u \ann U$.
On the other hand, kinds are inferred (synthesized) for type variables by looking them up in the
context $\Xi$, as well as for type-level applications, recursive types and stratified types.
Subtly, the kinding for type applications $T \App M$ requires that $T$ synthesize a kind: in
particular $T$ cannot be a type-level lambda, which would be checked against a kind.
This means that types in $\type$ do not arise from reducible lambda applications: lambdas must occur
within recursive or stratified types.
Finally, recursive and stratified types synthesize the kinds in their annotations.
Recursive type variables are added to the context $\Xi$ for checking the body of the type.
Stratified types require checking the constituent types using the index information gleaned in each
branch: $T_0$ is checked against $K\pr[\sub{0}{u}]$ and $T_s$ against $K\pr[\sub{\suc u}{u}]$.
}{
We employ a bidirectional kinding system to show when kinds can be inferred and
when kinding annotations are necessary.
The judgments use two contexts: an index context $\Delta$ and a type variable
context $\Xi$.
Note that in general the kinds assigned to type variables in $\Xi$ may
depend on index variables in $\Delta$.}

\begin{figureone}{
  Kinding rules for \tores
  \label{fig:kinding}
}
\[
\begin{array}{l}
\fbox{$\Delta;\Xi \gives T \chk K$}~\qquad\mbox{Check type $T$ against kind $K$}
\\[0.8em]
\infer{\Delta;\Xi \gives \Unit \chk \type}{}
\qquad
\infer{\Delta;\Xi \gives T_1 \cross T_2 \chk \type}{
       \Delta;\Xi \gives T_1 \chk \type & \Delta;\Xi \gives T_2 \chk \type}
\qquad
\infer{\Delta;\Xi \gives T_1 \tsum T_2 \chk \type}{
       \Delta;\Xi \gives T_1 \chk \type & \Delta;\Xi \gives T_2 \chk \type}
\\[1em]
\infer{\Delta;\Xi \gives \iarr{\many{u \ann U}}{S}{T} \chk \type}
      {\Delta \gives (\many{u \ann U}) \itype &
       \Delta,\many{u \ann U};\Xi \gives S \chk \type &
       \Delta,\many{u \ann U};\Xi \gives T \chk \type}
\qquad
\infer{\Delta;\Xi \gives \sigt{u \ann U}{T} \chk \type}{
       \Delta \gives U \itype & \Delta,u \ann U;\Xi \gives T \chk \type}
\\[1em]
\infer{\Delta;\Xi \gives M = N \chk \type}{
       \Delta \gives M \oft U & \Delta \gives N \oft U}
\qquad
\infer{\Delta;\Xi \gives \Lam u T \chk \pik{u \ann U} K}{\Delta,u \ann U;\Xi \gives T \chk K}
\qquad
 \infer{\Delta;\Xi \gives T \chk \type}{\Delta;\Xi \gives T \syn \type}
\\[1em]
\fbox{$\Delta;\Xi \gives T \syn K$}~\qquad\mbox{Infer a kind $K$ for type $T$}
\\[0.8em]
\infer{\Delta;\Xi \gives X \syn K}{X \ann K \in \Xi}
\qquad
\infer{\Delta;\Xi \gives T \App M \syn K[\sub M u]}{
       \Delta;\Xi \gives T \syn \pik{u \ann U} K &
       \Delta \gives M \oft U}
\\[1em]
\infer{\Delta;\Xi \gives \rect {X \ann K} T \syn K}
      {\Delta;\Xi,X \ann K \gives T \chk K}
\qquad
\infer{\Delta;\Xi \gives \corect {X \ann K} T \syn K}
      {\Delta;\Xi,X \ann K \gives T \chk K}
\\[1em]
\infer{\Delta;\Xi \gives \Rec_{K} (0 \mapsto T_0 \sep \suc u,\,X \mapsto T_s) \syn K}
      {K = \pik{u \ann \nat} K\pr &
      \Delta;\Xi \gives T_0 \chk K\pr[\sub{0}{u}] &
      \Delta,u \ann \nat;\Xi,X \ann K\pr \gives T_s \chk K\pr[\sub{\suc u}{u}]}
\\[1em]
\end{array}
\]
\end{figureone}

\subsection{Value Typing}

Values are the results of evaluation.
Note that values are closed, and hence their typing judgment does not require a
context.
However, closures do contain terms (typed with the main typing judgment) and
environments (typed against the contexts $\Delta$ and $\Gamma$).

\begin{figureone}{
  Value and environment typing
  \label{fig:valtyping}
}
\[
\begin{array}{l}
\fbox{$v : T$}\quad
\mbox{Value $v$ has type $T$} \\[0.8em]
\infer
    {\close{g}{\theta}{\sigma} : \msub \theta T} {
     \cdot \gives \theta : \Delta
  & \sigma : \msub{\theta} \Gamma
  & \Delta; \empt; \Gamma \gives g \chk T
    }
\quad
\infer{\unit : \Unit}{}
\quad
\infer{\refl : M = N}{\cdot \vdash M = N}
\quad
\infer{\pair{v_1}{v_2} : T_1 \times T_2}
      {v_1 : T_1 & v_2 : T_2}
\\[0.8em]
\infer{\inj{i} v : T_1 \tsum T_2}{v : T_i}
\quad
\infer{\pack{M}{v} : \sigt{u \ann U} T}
      {\cdot \gives M \oft U & v : T[\sub{M}{u}]}
\quad
\infer{\inj{\mu} v : (\rect{X \ann K} \Lam{\vec u} T) \App \vec M}
      {v : T[\many{\sub{M}{u}};\sub{\rect{X \ann K} \Lam{\vec u} T}{X}]}
\\[0.8em]
\infer{\inj{0} v \oft T_{\Rec} \App 0 \App \vec M}{v \oft T_0 \App \vec M}
\quad
\infer{\inj{\suc} v \oft T_{\Rec} \App (\suc N) \vec M}{v \oft T_s[\sub{N}{u};\sub{T_{\Rec} \App N}{X}] \App \vec M}
\\[0.8em]
\infer
    {\close{\corec f t}{\theta}{\sigma} \cdot \vec N\ v
      : \msub \theta {(\corect{X \ann K} \Lam{\many{u\pr}} T)}\ \vec N}
    {\deduce{\Delta;\cdot,X \ann K;\Gamma,f \ann (\iarr{\many{u \ann U}}{S}{X \vec{u}})
     \gives t \chk \iarr{\many{u \ann U}}{S}{T[\many{\sub{u}{u\pr}}]}}
     {\cdot \gives \theta : \Delta
  & \sigma : \msub{\theta} \Gamma
  & \cdot \vdash \vec N : \vec U
  & v : S[\theta,\many{\sub{N}{u}}]}}
\\[1em]
\fbox{$\sigma: \Gamma$}\quad
\mbox{Environment $\sigma$ has domain $\Gamma$} \\[0.8em]
    \infer{\cdot : \cdot}{}
\qquad
\infer{(\sigma, v/x) : \Gamma, x{:}T}
      {\sigma : \Gamma \qquad v : T}
\end{array}
\]
\end{figureone}

\LONGVERSION{
\subsection{Proofs}

\subred
\begin{proof}
  By mutual induction on the evaluation judgments.
  We include a few key cases here.
  \input{subred.tex}
\end{proof}

\prefp
\begin{proof}
  Recall that $\semlfp \F = \meet \mc S$ where
  $\mc S = \{ \C \in \mc L \sep \fall{\X \in \mc L} \X \leq \C \implies \F(\X) \leq \C \}$.
  To show $\F(\semlfp \F) \leq \meet \mc S$, it suffices to show
  $\F(\semlfp \F) \leq \C$ for every $\C \in \mc S$.
  By definition of the meet, $\semlfp \F \leq \C$ for every $\C \in \mc S$.
  But by definition of $\mc S$, this implies that $\F(\semlfp \F) \leq \C$ as required.
  (This argument exploits our impredicative definition of $\semlfp$.)
\end{proof}

\funprefp
\begin{proof}
  We will reframe the lemma statement using a new piece of notation.
  For a closure $c \in \VAL$ and $\B \in \vec \U \to \powset \VAL = \mc L$, define $\E_c(\B) \in \mc L$ by
  $\E_c(\B)(\vec M) = \{ v \in \VAL \sep \evalapp{c}{\vec M}{v}{w} \in \B(\vec M) \}$.
  One can see that $c \in \semarr{\vec \U}{\A}{\B} \iff \A \leq \E_c(\B)$ (using the ordering on $\mc L$).
  We can now rewrite the lemma as the following:
  \begin{enumerate}
  \item if $\fall{\X \in \mc L} \X \leq \E_c(\B) \implies \F \X \leq \E_c(\B)$ then $\semlfp \F \leq \E_c(\B)$;
  \item if $\fall{\X \in \mc L} \E_c(\B) \leq \X \implies \E_c(\B) \leq \F \X$ then $\E_c(\B) \leq \semgfp \F$.
  \end{enumerate}

  To prove each of them, we first assume the premise. Let us do the first statement now.
  Recall that $\semlfp \F = \meet \mc S$ where
  $\mc S = \{ \C \in \mc L \sep \fall{\X \in \mc L} \X \leq \C \implies \F(\X) \leq \C \}$.
  By definition of the meet, $\semlfp \F \leq \C$ for every $\C \in \mc S$.
  Therefore it suffices to show that there is just one $\C \in \mc S$ for which $\C \leq \E_c(\B)$.
  However, our assumption is exactly that $\E_c(\B) \in \mc S$, and clearly $\E_c(\B) \leq \E_c(\B)$,
  so we are done.
  (This proof again makes use of impredicativity in the definition of $\semlfp$.)

  The second case follows the same idea. $\semgfp \F = \join \mc S$ where
  $\mc S = \{ \C \in \mc L \sep \fall{\X \in \mc L} \C \leq \X \implies \C \leq \F(\X) \}$.
  By definition of join, $\C \leq \semgfp F$ for every $\C \in \mc S$. It thus suffices
  to show that there is one $\C \in \mc S$ for which $\E_c(\B) \leq \C$. Again,
  $\E_c(\B) \in \mc S$ and $\E_c(\B) \leq \E_c(\B)$ so we are done.
\end{proof}

\appprefp
\begin{proof}
  Let $\mc L = \sem{K}(\theta)$. \\
  Define $\F : \mc L \to \mc L$ by $\F(\X) = \vec N \mapsto \inj{\mu} \sem{S}(\theta,\many{\sub{N}{u}};\eta,\sub{\X}{X})$. \\
  Then $\sem{R}(\theta;\eta)$ \\
  $= \semlfp(\X \mapsto \injop{\mu} \sem{\Lam{\vec u} S}(\theta;\eta,\sub{\X}{X}))$
    \hfill by $\sem{\rect X T}$ def. \\
  $= \semlfp(\X \mapsto \vec N \mapsto \inj{\mu} \sem{S}(\theta,\many{\sub{N}{u}};\eta,\sub{\X}{X}))$
  \hfill by $\sem{\Lam{\vec u} T}$ def. \\
  $= \semlfp \F$ \hfill by $\F$ def. \\
  $\F(\sem{R}(\theta;\eta)) \leq_{\mc L} \sem{R}(\theta;\eta)$ \hfill by Lemma \ref{lem:prefp} \\
  Now $\inj{\mu} \sem{S[\many{\sub{M}{u}};\sub{R}{X}]}(\theta;\eta)$ \\
  $= \inj{\mu} \sem{S}((\id_{\Delta},\many{\sub{M}{u}})[\theta];\eta,\sub{\sem{R}(\theta;\eta)}{X})$
    \hfill by Lemma \ref{lem:subst} \\
  $= \inj{\mu} \sem{S}(\theta,\many{\sub{M[\theta]}{u}};\eta,\sub{\sem{R}(\theta;\eta)}{X})$
    \hfill by \LONGSHORT{Def. \ref{def:isubst-comp}}{composition of index substitutions} \\
  $= \F(\sem{R}(\theta;\eta))(\many{M[\theta]})$ \hfill by $\F$ def. \\
  $\subseteq \sem{R}(\theta;\eta)(\many{M[\theta]})$
    \hfill since $\F(\sem{R}(\theta;\eta)) \leq_{\mc L} \sem{R}(\theta;\eta)$ \\
  $= \sem{R \App \vec M}(\theta;\eta)$ \hfill by $\sem{R \App \vec M}$ def. \\
\end{proof}

\backclos
\begin{proof}
  \begin{enumerate}
  \item Let $c = (\rec f t)[\theta;\sigma]$. \\
  Suppose $\vec M \in \vec \U$ and $v\pr \in (\injop{\mu} \A)(\vec M)$. \\
  $v \in \inj{\mu} \A(\vec M)$ \hfill by $\injop{}$ def. \\
  $v\pr = \inj{\mu} v$ where $v \in \A(\vec M)$ \\
  $t[\theta;\sigma, \sub{c}{f}] \evto c\pr \in \semarr{\vec \U}{\A}{\B}$ \hfill assumption of lemma \\
  $\evalapp{c\pr}{\vec M}{v}{w} \in \B(\vec M)$ \hfill by $\semarr{\vec \U}{\A}{\B}$ def. \\
  $\evalapp{c}{\vec M}{(\inj{\mu} v)}{w} \in \B(\vec M)$ \hfill by $\eapprec$ \\
  $\evalapp{c}{\vec M}{v\pr}{w} \in \B(\vec M)$ \hfill since $v\pr = \inj{\mu} v$ \\
  $c \in \semarr{\vec \U}{\injop{\mu} \A}{\B}$ \hfill by $\semarr{\vec \U}{\A}{\B}$ def.\\

  \item Let $c = (\corec f t)[\theta;\sigma]$. \\
    Suppose $\vec M \in \vec \U$ and $v \in \A(\vec M)$.\\
    $t[\theta;\sigma, \sub{c}{f}] \evto c\pr \in \semarr{\vec \U}{\A}{\B}$ \hfill assumption \\
    $\evalapp{c\pr}{\vec M}{v}{w} \in \B(\vec M)$ \hfill by $\semarr{\vec \U}{\A}{\B}$ def. \\
    $\appout{(c \cdot \vec M\ v)} \evto w \in \B(\vec M)$ \hfill by $\eout{\nu}$\\
    $c \cdot  \vec M\ v \in \outop{\nu}(\B(\vec M))$ \hfill by $\outop{\nu}$ def. \\
    $c \in \semarr{\vec \U}{\A}{\outop{\nu} \B}$ \hfill by $\semarr{\vec \U}{\A}{\B}$ def.
  \end{enumerate}
\end{proof}

\begin{lemma}[Stratified types equivalent to unfolding]
\label{lem:trec}$\;$\\
  Let $T_{\Rec} \equiv \Recnat{T_z}{n}{X}{T_s}{K}$ where
  $K = \pik{n \ann \nat} \pik{\many{u \ann U}} \type$ and $\Delta;\Xi \gives T_{\Rec} \syn K$,
  and $\Delta \gives \vec M \oft (\many{u \ann U})$ and $\Delta \gives N \oft \nat$
  and $\gives \theta \oft \Delta$ and $\eta \in \sem{\Xi}(\theta)$. Then
  \begin{enumerate}
    \item
    $\sem{T_{\Rec} \App 0 \App \vec M}(\theta;\eta) = \inj{0} (\sem{T_z \App \vec M}(\theta;\eta))$ and
    \item
    $\sem{T_{\Rec} \App (\suc N) \App \vec M}(\theta;\eta) =
    \inj{\suc} (\sem{T_s[\sub{N}{n};\sub{(T_{\Rec} \App N)}{X}] \App \vec M}(\theta;\eta))$.
  \end{enumerate}
\end{lemma}

\begin{proof}
  Let $\C = \injop{0} \sem{T_z}(\theta;\eta)$ and
  $\F = (N \mapsto \X \mapsto \injop{\suc} \sem{T_s}(\theta,\sub{N}{n};\eta,\sub{\X}{X}))$.
  \begin{enumerate}
    \item
      $\sem{T_{\Rec} \App 0 \App \vec M}(\theta;\eta)$ \\
      $= \sem{T_{\Rec}}(\theta;\eta)(0)(\many{M[\theta]})$ \hfill by $\sem{T \App \vec M}$ def. \\
      $= \RECnat{\C}{\F} \App 0 \App \many{M[\theta]}$ \hfill by $\sem{T_{\Rec}}$ def. \\
      $= \C \App \many{M[\theta]}$ \hfill by $\REC$ def \\
      $= (\injop{0} \sem{T_z}(\theta;\eta))(\many{M[\theta]})$ \hfill by $\C$ def. \\
      $= \inj{0} (\sem{T_z}(\theta;\eta)(\many{M[\theta]}))$ \hfill by $\injop{}$ def. \\
      $= \inj{0} (\sem{T_z \App \vec M}(\theta;\eta))$ \hfill by $\sem{T \App \vec M}$ def. \\
    \item
      $\sem{T_{\Rec} \App (\suc N) \App \vec M}(\theta;\eta)$ \\
      $= \sem{T_{\Rec}}(\theta;\eta)(\suc N[\theta])(\many{M[\theta]})$ \hfill by $\sem{T \App \vec M}$ def. \\
      $= \RECnat{\C}{\F} \App (\suc N[\theta]) \App \many{M[\theta]}$ \hfill by $\sem{T_{\Rec}}$ def. \\
      $= \F \App N[\theta] \App (\RECnat{\C}{\F} \App N[\theta]) \App \many{M[\theta]}$ \hfill by $\REC$ def. \\
      $= (\injop{\suc} \sem{T_s}(\theta,\sub{N[\theta]}{n};\eta,\sub{(\RECnat{\C}{\F} \App N[\theta])}{X}))(\many{M[\theta]})$ \hfill by $\F$ def. \\
      $= \inj{\suc} (\sem{T_s}(\theta,\sub{N[\theta]}{n};\eta,\sub{(\RECnat{\C}{\F} \App N[\theta])}{X})(\many{M[\theta]}))$
      \hfill by $\injop{}$ def. \\
      $= \inj{\suc} (\sem{T_s}((\id_{\Delta},\sub{N}{n})[\theta];\eta,\sub{\sem{T_{\Rec} \App N}(\theta;\eta)}{X})(\many{M[\theta]}))$
      \hfill by \LONGSHORT{Def. \ref{def:isubst-comp}}{composition of index substitutions} and $\sem{T_{\Rec}}$ def. \\
      $= \inj{\suc} (\sem{T_s[\sub{N}{n};\sub{(T_{\Rec} \App N)}{X}]}(\theta;\eta)(\many{M[\theta]}))$ \hfill by Lemma \ref{lem:subst} \\
      $= \inj{\suc} (\sem{T_s[\sub{N}{n};\sub{(T_{\Rec} \App N)}{X}] \App \vec M}(\theta;\eta))$ \hfill by $\sem{T \App \vec M}$ def. \\
  \end{enumerate}
\end{proof}

\termination
\begin{proof}
  The proof is by induction on the typing derivation.
  Technically this is a mutual induction on the dual judgments of type checking and synthesis.
  In each case we introduce the assumptions $\gives \theta \oft \Delta$ and $\eta \in \sem{\Xi}(\theta)$
  and $\sigma \in \sem{\Gamma}(\theta;\eta)$, where $\Delta$, $\Xi$ and $\Gamma$ appear in the
  conclusion of the relevant typing rule.
  We will slightly abuse notation to introduce an existentially quantified variable in the judgment $t[\theta;\sigma] \evto v \in \V$, to mean that
  $\exists v.\, t[\theta;\sigma] \evto v \wedge v \in \V$.
  Note that the cases involving stratified types and induction over indices are specific to the
  particular index language (and assume an induction principle over closed index types).
  The rest of the proof is generic, only assuming the properties in Section \ref{sec:tores-index}.

  \paragraph*{Case:}
  \[ \infer[\tunit]{\Delta;\Xi;\Gamma \gives \unit \chk \Unit}{} \]
  $\unit[\theta;\sigma] \evto \unit$ \hfill by $\eunit$ \\
  $\unit \in \sem{\Unit}(\theta;\eta)$ \hfill by $\sem \Unit$ def. \\

  \paragraph*{Case:}
  \[ \infer[\tvar]{\Delta;\Xi;\Gamma \gives x \syn T}{x \ann T \in \Gamma} \]
  $\sigma \in \sem{\Gamma}(\theta;\eta)$ \hfill assumption of Thm \ref{thm:term} \\
  $x \ann T \in \Gamma$ \hfill premise of $\tvar$ \\
  $\sigma(x) = v \in \sem{T}(\theta;\eta)$ \hfill by Def. \ref{def:sem_typctx} \\
  $x[\theta;\sigma] \evto v$ \hfill by $\evar$ \\

  \paragraph*{Case:}
  \[ \infer[\tlam]
       {\Delta;\Xi;\Gamma \gives \ilam{\vec u}{x} s \chk \iarr{\many{u \ann U}}{R}{S}}
       {\Delta,\many{u \ann U};\Xi;\Gamma,x \ann R \gives s \chk S} \]
  Let $c$ be the closure $(\ilam{\vec u}{x}{s})[\theta;\sigma]$. \\
  $(\ilam{\vec u}{x}{s})[\theta;\sigma] \evto c$ \hfill by $\elam$ \\
  Suffices to show $c \in \sem{\iarr{\many{u \ann U}}{R}{S}}(\theta;\eta)$, i.e. \\
  $\fall{\vec M \in \sem{(\many{u \ann U})}(\theta)} \fall{v \in \sem{R}(\theta\pr;\eta)}
   \evalapp{c}{\vec M}{v}{w} \in \sem{S}(\theta\pr;\eta)$,
  where $\theta\pr = \theta,\many{\sub{M}{u}}$. \\
  Suppose $\vec M \in \sem{(\many{u \ann U})}(\theta)$ and $v \in \sem{R}(\theta\pr;\eta)$
  where $\theta\pr = \theta,\many{\sub{M}{u}}$. \\
  $\gives \theta \oft \Delta$ \hfill assumption of Thm \ref{thm:term} \\
  $\gives \theta\pr \oft \Delta,\many{u \ann U}$ \hfill by index substitution typing \\
  Let $\sigma\pr = \sigma,\sub{v}{x}$. \\
  $\sigma \in \sem{\Gamma}(\theta;\eta)$ \hfill assumption of Thm \ref{thm:term} \\
  $\sigma \in \sem{\Gamma}(\theta\pr;\eta)$ \hfill since $\vec u \notin \FV(\Gamma)$ \\
  $\sigma\pr \in \sem{\Gamma,x \ann R}(\theta\pr;\eta)$ \hfill by Def. \ref{def:sem_typctx} \\
  $s[\theta\pr;\sigma\pr] \evto w \in \sem{S}(\theta\pr;\eta)$ \hfill
  by I.H. with $\theta\pr$, $\eta$ and $\sigma\pr$ \\
  $\evalapp{c}{\vec M}{v}{w}$ \hfill by $\eapplam$ \\

  \paragraph*{Case:}
  \[ \infer[\tapp]{\Delta;\Xi;\Gamma \gives q \app \vec M \app r \syn S[\many{M/u}]}
                  {\Delta;\Xi;\Gamma \gives q \syn \iarr{\many{u \ann U}}{R}{S} &
                   \Delta \gives \vec M \oft (\many{u \ann U}) &
                   \Delta;\Xi;\Gamma \gives r \chk R[\many{\sub{M}{u}}]} \]
  $q[\theta;\sigma] \evto c \in \sem{\iarr{\many{u \ann U}}{R}{S}}(\theta;\eta)$ \hfill by I.H. \\
  $c \in \semarr{\sem{(\many{u \ann U})}(\theta)}
                {(\vec N \mapsto \sem{R}(\theta,\many{\sub{M}{u}};\eta))}
                {(\vec N \mapsto \sem{S}(\theta,\many{\sub{M}{u}};\eta))}$
  \hfill by $\sem{\iarr{\many{u \ann U}}{R}{S}}(\theta;\eta)$ def. \\
  $\many{M[\theta]} \in \sem{(\many{u \ann U})}(\theta)$ \hfill by Lemma \ref{lem:sem_isubst_spine} \\
  $r[\theta;\sigma] \evto v \in \sem{R[\many{\sub{M}{u}}]}(\theta;\eta)$ \hfill by I.H. \\
  $v \in \sem{R}((\id_{\Delta},\many{\sub{M}{u}})[\theta];\eta)$ \hfill by Lemma \ref{lem:subst} \\
  $v \in \sem{R}(\theta,\many{\sub{M[\theta]}{u}};\eta)$
  \hfill by \LONGSHORT{Def. \ref{def:isubst-comp}}{composition of index substitutions} \\
  $\appval{c}{\many{M[\theta]}}{v} \evto w \in \sem{S}(\theta,\many{\sub{M[\theta]}{u}};\eta)$
  \hfill by \cref{def:sem-funsp} \\
  $w \in \sem{S}((\id_{\Delta},\many{\sub{M}{u}})[\theta];\eta)$
  \hfill by \LONGSHORT{Def. \ref{def:isubst-comp}}{composition of index substitutions} \\
  $w \in \sem{S[\many{\sub{M}{u}}]}(\theta;\eta)$ \hfill by Lemma \ref{lem:subst} \\
  $(q \app \vec M \app r)[\theta;\sigma] \evto w$ \hfill by $\eapp$ \\

  \paragraph*{Case:}
  \[ \infer[\tpair]{\Delta;\Xi;\Gamma \gives \pair{t_1}{t_2} \chk T_1 \times T_2}
                   {\Delta;\Xi;\Gamma \gives t_1 \chk T_1 & \Delta;\Xi;\Gamma \gives t_2 \chk T_2} \]
  $t_i[\theta;\sigma] \evto v_i \in \sem{T_i}(\theta;\eta)$ for $i \in \{1, 2\}$ \hfill by I.H. \\
  $\pair{t_1}{t_2}[\theta;\sigma] \evto \pair{v_1}{v_2}$ \hfill by $\epair$ \\
  $\pair{v_1}{v_2} \in \sem{T_1 \cross T_2}(\theta;\eta)$ \hfill by $\sem{T_1 \cross T_2}$ def. \\

  \paragraph*{Case:}
  \[ \infer[\tsplit]{\Delta;\Xi;\Gamma \gives \psplit{p}{x_1}{x_2} s \chk T}
                    {\Delta;\Xi;\Gamma \gives p \syn T_1 \times T_2 &
                     \Delta;\Xi;\Gamma,x_1 \ann T_1,x_2 \ann T_2 \gives s \chk T} \]
  $p[\theta;\sigma] \evto w \in \sem{T_1 \cross T_2}(\theta;\eta)$ \hfill by I.H. \\
  $w = \pair{v_1}{v_2}$ where $v_i \in \sem{T_i}(\theta;\eta)$ for $i \in \{1, 2\}$ \hfill by $\sem{T_1 \cross T_2}$ def. \\
  $\sigma,\sub{v_1}{x_1},\sub{v_2}{x_2} \in \sem{\Gamma,x_1 \ann T_1,x_2 \ann T_2}(\theta;\eta)$ \hfill by Def. \ref{def:sem_typctx} \\
  $s[\theta;\sigma,\sub{v_1}{x_1},\sub{v_2}{x_2}] \evto v \in \sem{T}(\theta;\eta)$ \hfill by I.H. \\
  $t[\theta;\sigma] \evto v$ \hfill by $\esplit$ \\

  \paragraph*{Case:}
  \[ \infer[\tinj{i}]{\Delta;\Xi;\Gamma \gives \inj{i} t_i \chk T_1 \tsum T_2}
                     {\Delta;\Xi;\Gamma \gives t_i \chk T_i}
     \qquad\text{for $i \in \{1, 2\}$} \]
  This is really two cases. Fix $i \in \{1, 2\}$. \\
  $t_i[\theta;\sigma] \evto v_i \in \sem{T_i}(\theta;\eta)$ \hfill by I.H. \\
  $(\inj{i} t_i)[\theta;\sigma] \evto \inj{i} v_i$ \hfill by $\einj{i}$. \\
  $\inj{i} v_i \in \inj{i} \sem{T_i}(\theta;\eta) \subseteq \sem{T_1 \tsum T_2}(\theta;\eta)$
  \hfill by $\sem{T_1 \tsum T_2}$ def. \\

  \paragraph*{Case:}
  \[ \infer[\tcase]{\Delta;\Xi;\Gamma \gives (\case s \inj{1} x_1 \mapsto t_1 \sep \inj{2} x_2 \mapsto t_2) \chk T}
                   {\Delta;\Xi;\Gamma \gives s \syn T_1 \tsum T_2 &
                    \Delta;\Xi;\Gamma,x_1 \ann T_1 \gives t_1 \chk T &
                    \Delta;\Xi;\Gamma,x_2 \ann T_2 \gives t_2 \chk T} \]
  $s[\theta;\sigma] \evto w \in \sem{T_1 \tsum T_2}(\theta;\eta)$ \hfill by I.H. \\
  $w \in \inj{1}\sem{T_1}(\theta;\eta)$ or $w \in \inj{2} \sem{T_2}(\theta;\eta)$
  \hfill by $\sem{T_1 \tsum T_2}$ def. \\
  $w = \inj{i} v_i$ where $v_i \in \sem{T_i}(\theta;\eta)$, for some $i \in \{1, 2\}$. \\
  $\sigma \in \sem{\Gamma}(\theta;\eta)$ \hfill by assumption of Thm \ref{thm:term} \\
  $\sigma,\sub{v_i}{x_i} \in \sem{\Gamma,x_i \ann T_i}(\theta;\eta)$ \hfill by Def. \ref{def:sem_typctx} \\
  $t_i[\theta;\sigma,\sub{v_i}{x_i}] \evto v \in \sem{T}(\theta;\eta)$ \hfill by I.H. on $t_i$ \\
  $t[\theta;\sigma] \evto v$ \hfill by $\ecase{i}$ \\

  \paragraph*{Case:}
  \[ \infer[\tpack]{\Delta;\Xi;\Gamma \gives \pack{M}{s} \chk \sigt{u \ann U} S}
                   {\Delta \gives M \oft U & \Delta;\Xi;\Gamma \gives s \chk S[\sub{M}{u}]} \]
  $s[\theta;\sigma] \evto w \in \sem{S[\sub{M}{u}]}(\theta;\eta)$ \hfill by I.H. \\
  $w \in \sem{S}((\id_{\Delta},\sub{M}{u})[\theta];\eta)$ \hfill by Lemma \ref{lem:subst} \\
  $w \in \sem{S}(\theta,\sub{M[\theta]}{u};\eta)$
  \hfill by \LONGSHORT{Def. \ref{def:isubst-comp}}{composition of index substitutions} \\
  $(\pack{M}{s})[\theta;\sigma] \evto \pack{M[\theta]}{w}$ \hfill by $\epack$ \\
  $M[\theta] \in \sem{U}(\theta)$ \hfill by Lemma \ref{lem:sem_isubst} \\
  $\pack{M[\theta]}{w} \in \sem{\sigt{u \ann U} S}(\theta;\eta)$ \hfill by $\sem{\sigt{u \ann U} S}$ def. \\

  \paragraph*{Case:}
  \[ \infer[\tunpack]{\Delta;\Xi;\Gamma \gives \unpack{p}{u}{x} q \chk T}
                     {\Delta;\Xi;\Gamma \gives p \syn \sigt{u \ann U}{S} &
                      \Delta,u \ann U;\Xi;\Gamma,x \ann S \gives q \chk T} \]
  $p[\theta;\sigma] \evto w\pr \in \sem{\sigt{u \ann U} S}(\theta;\eta)$ \hfill by I.H. \\
  $w\pr = \pack{M}{w}$ where $M \in \sem{U}(\theta)$ and $w \in \sem{S}(\theta,\sub{M}{u};\eta)$
  \hfill by $\sem{\sigt{u \ann U} S}$ def. \\
  Let $\theta\pr = \theta,\sub{M}{u}$. \\
  $\gives \theta\pr \oft \Delta,u \ann U$ \hfill by index substitution typing \\
  $\sigma \in \sem{\Gamma}(\theta;\eta)$ \hfill assumption of Thm \ref{thm:term} \\
  $\sigma \in \sem{\Gamma}(\theta\pr;\eta)$ \hfill since $u \notin \FV(\Gamma)$ \\
  Let $\sigma\pr = \sigma,\sub{w}{x}$. \\
  $\sigma\pr \in \sem{\Gamma,x \ann S}(\theta\pr;\eta)$ \hfill by Def. \ref{def:sem_typctx} \\
  $q[\theta\pr;\sigma\pr] \evto v \in \sem{T}(\theta\pr;\eta)$ \hfill by I.H. \\
  $(\unpack{p}{u}{x}{q})[\theta;\sigma] \evto v$ \hfill by $\eunpack$ \\
  $v \in \sem{T}(\theta;\eta)$ \hfill since $u \notin \FV(T)$ \\

  \paragraph*{Case:}
  \[ \infer[\trefl]{\Delta;\Xi;\Gamma \gives \refl \chk M = N}{\Delta \gives M = N} \]
  $\gives M[\theta] = N[\theta]$ \hfill by Req. \ref{req:isubst-eq} \\
  $\refl \in \sem{M = N}(\theta;\eta)$ \hfill by $\sem{M = N}$ def. \\
  $\refl[\theta;\sigma] \evto \refl$ \hfill by $\erefl$ \\

  \paragraph*{Case:}
  \[ \infer[\teq]{\Delta;\Xi;\Gamma \gives \eqelim{q}{\Theta}{\Delta\pr}{s} \chk T}
                 {\Delta;\Xi;\Gamma \gives q \syn M = N &
                  \Delta \gives M \unif N \gen (\Delta\pr \mid \Theta) &
                  \Delta\pr;\Xi[\Theta];\Gamma[\Theta] \gives s \chk T[\Theta]} \]
  $q[\theta;\sigma] \evto w \in \sem{M = N}(\theta;\eta)$ \hfill by I.H. \\
  $w = \refl$ and $\gives M[\theta] = N[\theta]$ \hfill by $\sem{M = N}$ def. \\
  $\Delta \gives M \unif N \gen (\Delta\pr \mid \Theta)$ \hfill premise of $\teq$ \\
  $\gives \theta \oft \Delta$ \hfill assumption of Thm \ref{thm:term} \\
  $\Delta\pr \gives \Theta \match \theta \synth (\empt \sep \theta\pr)$ \hfill
  \hfill by \cref{req:match-comp} \\
  $\sem{\Xi[\Theta]}(\theta\pr) = \sem{\Xi}(\Theta[\theta\pr])$ and
  $\sem{T[\Theta]}(\theta\pr;\eta) = \sem{T}(\Theta[\theta\pr];\eta)$ \hfill by Lemma \ref{lem:subst} \\
  $\sem{\Xi[\Theta]}(\theta\pr) = \sem{\Xi}(\theta])$ and
  $\sem{T[\Theta]}(\theta\pr;\eta) = \sem{T}(\theta;\eta)$ \hfill by Thm \ref{thm:sound-match} \\
  Extending Lemma \ref{lem:subst} from types $T$ to typing contexts $\Gamma$, \\
  $\sem{\Gamma[\Theta]}(\theta\pr;\eta) = \sem{\Gamma}(\Theta[\theta\pr];\eta) = \sem{\Gamma}(\theta;\eta)$. \\
  $\eta \in \sem{\Xi[\Theta]}(\theta\pr)$ \hfill since $\eta \in \sem{\Xi}(\theta)$ \\
  $\sigma \in \sem{\Gamma[\Theta]}(\theta\pr;\eta)$ \hfill since $\sigma \in \sem{\Gamma}(\theta;\eta)$ \\
  $s[\theta\pr;\sigma] \evto v \in \sem{T[\Theta]}(\theta\pr;\eta)$ \hfill by I.H. \\
  $v \in \sem{T}(\theta;\eta)$ \\
  $(\eqelim{q}{\Theta}{\Delta\pr}{s})[\theta;\sigma] \evto v$ \hfill by $\eeq$ \\

  \paragraph*{Case:}
  \[ \infer[\teqfalse]{\Delta;\Xi;\Gamma \gives \eqelimfalse s \chk T}
                      {\Delta;\Xi;\Gamma \gives s \syn M = N &
                       \Delta \gives M \unif N \gen \fail} \]
  $s[\theta;\sigma] \evto w \in \sem{M = N}(\theta;\eta)$ \hfill by I.H. \\
  $w = \refl$ and $\gives M[\theta] = N[\theta]$ \hfill by $\sem{M = N}$ def. \\
  $\theta$ unifies $M$ and $N$ in $\Delta$ \hfill by def. of a unifier \\
  $\Delta \gives M \unif N \gen \fail$ \hfill premise of $\teqfalse$ \\
  There is no unifier for $M$ and $N$ \hfill consequence of \cref{req:dec-unif} \\
  Contradiction: derive $(\eqelimfalse s)[\theta;\sigma] \evto v \in \sem{T}(\theta;\eta)$ \\

  \paragraph*{Case:}
  \[ \infer[\tinj{\mu}]{\Delta;\Xi;\Gamma \gives \inj{\mu} s \chk (\rect{X \ann K} \Lam{\vec u} S) \App \vec M}
                       {\Delta;\Xi;\Gamma \gives s \chk S[\many{\sub{M}{u}};\sub{\rect{X \ann K} \Lam{\vec u} S}{X}]} \]
  Let $R = \rect{X \ann K} \Lam{\vec u} S$. \\
  $s[\theta;\sigma] \evto w \in \sem{S[\many{\sub{M}{u}};\sub{R}{X}]}(\theta;\eta)$ \hfill by I.H. \\
  $(\inj{\mu} s)[\theta;\sigma] \evto \inj{\mu} w$ \hfill by $\einj{\mu}$ \\
  $\inj{\mu} w \in \sem{R \App \vec M}(\theta;\eta)$ \hfill by Lemma \ref{lem:app_prefp} \\


  \paragraph*{Case:}
  \[ \infer[\trec]{\Delta;\Xi;\Gamma \gives \rec f s \chk \iarr{\many{u \ann U}}{(\rect{X \ann K} \Lam{\many{u\pr}} R) \App \vec{u}}{S}}
                  {\Delta;\Xi,X \ann K;\Gamma, f \ann \iarr{\many{u \ann U}}{X \App \vec u}{S} \gives
                     s \chk \iarr{\many{u \ann U}}{R[\many{u/u\pr}]}{S}
                   & \vec u \notin \FV(R)} \]
  Let $c = (\rec f s)[\theta;\sigma]$ and
  $C = (\rect{X \ann K} \Lam{\many{u\pr}} R) \App \vec{u}$.
  By $\erec$, $(\rec f s)[\theta;\sigma] \evto c$.
  We need to show that $c \in \sem{\iarr{\many{u \ann U}}{C}{S}}(\theta;\eta)$,

  Let $\vec \U = \sem{(\many{u \ann U})}(\theta) = \sem{(\many{u\pr \ann U})}(\theta)$
  (both $\vec u$ and $\many{u\pr}$ do not appear in $\theta$).
  From the kinding rules we know that $K = \pik{\many{u\pr \ann U}} \type$ so
  $\sem{K}(\theta) = \vec \U \to \powset \VAL$.
  Let $\mc L = \sem{K}(\theta)$ and define $\F \in \mc L \to \mc L$ by
  $\F(\X) = \injop{\mu} \sem{\Lam{\many{u\pr}} R}(\theta;\eta,\sub{\X}{X})$.

  For $\vec M \in \vec \U$,
  \begin{align*}
    \sem{C}(\theta,\many{\sub{M}{u}};\eta)
    &= \sem{\rect{X \ann K} \Lam{\many{u\pr}} R}(\theta,\many{\sub{M}{u}};\eta)(\vec u[\theta,\many{\sub{M}{u}}])
      &\qquad\text{by \sem{T} def} \\
    &= \sem{\rect{X \ann K} \Lam{\many{u\pr}} R}(\theta;\eta)(\vec M) \\
      &\qquad\text{dropping mappings for $\vec u$ since $\vec u \notin \FV(R)$} \\
    &= \semlfp(\X \mapsto \injop{\mu} \sem{\Lam{\many{u\pr}} R}(\theta;\eta,\sub{\X}{X}))(\vec M)
      &\qquad\text{by \sem{T} def} \\
    &= (\semlfp \F)(\vec M) &\qquad\text{by $\F$ def.} \\
  \end{align*}

  Define $\B \in \mc L$ by $\B(\vec M) = \sem{S}(\theta,\many{\sub{M}{u}};\eta)$.
  Then
  \[ \sem{\iarr{\many{u \ann U}}{C}{S}}(\theta;\eta) = \semarr{\vec \U}{\semlfp \F}{\B} \qquad\text{by $\sem T$ def.} \]

  We want to show $c \in \semarr{\vec \U}{\semlfp \F}{\B}$.
  We will instead prove the sufficient condition given in Lemma \ref{lem:fun-prefp}.
  To this end, suppose $\X \in \vec \U \to \powset \VAL = \mc L$ and $c \in \semarr{\vec \U}{\X}{\B}$.
  The goal is now to show $c \in \semarr{\vec \U}{\F(\X)}{\B}$.

  Define $\A \in \mc L$ by $\A = \sem{\Lam{\many{u\pr}} R}(\theta;\eta,\sub{\X}{X})$,
  so $\F(\X) = \injop{\mu} \A$.
  By Lemma \ref{lem:back_clos}, it suffices to show that
  \[ s[\theta;\sigma,\sub{c}{f}] \evto c\pr \in \semarr{\vec \U}{\A}{\B}. \]

  We need to interpret the types $\iarr{\many{u \ann U}}{X \App \vec u}{S}$ and
  $\iarr{\many{u \ann U}}{R[\many{\sub{u}{u\pr}}]}{S}$ appearing in the premise of $\trec$.
  Note that these types are well-kinded under the contexts $\Delta;\Xi,X \ann K$.
  Since $\X \in \mc L = \sem{K}(\theta)$, we interpret them under the environments $\theta$ and
  $\eta,\sub{\X}{X} \in \sem{\Xi,X \ann K}(\theta)$.

  \begin{align*}
    &\sem{\iarr{\many{u \ann U}}{X \App \vec u}{S}}(\theta;\eta,\sub{\X}{X}) \\
    &= \semarr{\vec \U}{(\vec M \mapsto \sem{X \App \vec u}(\theta,\many{\sub{M}{u}};\eta,\sub{\X}{X}))}
                             {(\vec M \mapsto \sem{S}(\theta,\many{\sub{M}{u}};\eta,\sub{\X}{X}))} \\
    &= \semarr{\vec \U}{(\vec M \mapsto \X(\vec M))}
                             {(\vec M \mapsto \sem{S}(\theta,\many{\sub{M}{u}};\eta))} \\
    &\qquad\text{dropping a mapping for $X$ since $X \notin \FV(S)$} \\
    &= \semarr{\vec \U}{\X}{\B}
  \end{align*}

  \begin{align*}
    &\sem{\iarr{\many{u \ann U}}{R[\many{\sub{u}{u\pr}}]}{S}}(\theta;\eta,\X/X) \\
    &= \semarr{\vec \U}{(\vec M \mapsto \sem{R[\many{\sub{u}{u\pr}}]}(\theta,\many{\sub{M}{u}};\eta,\sub{\X}{X}))}
                             {(\vec M \mapsto \sem{S}(\theta,\many{\sub{M}{u}};\eta,\sub{\X}{X}))} \\
    &= \semarr{\vec \U}{(\vec M \mapsto \sem{R}((\id_{\Delta},\many{\sub{u}{u\pr}})[\theta,\many{\sub{M}{u\pr}}];\eta,\sub{\X}{X}))}
                             {(\vec M \mapsto \sem{S}(\theta,\many{\sub{M}{u}};\eta))} \\
    &\qquad\text{by Lemma \ref{lem:subst} and again dropping a mapping for $X$} \\
    &= \semarr{\vec \U}{(\vec M \mapsto \sem{R}(\theta,\many{\sub{M}{u\pr}};\eta,\sub{\X}{X}))}
                             {(\vec M \mapsto \sem{S}(\theta,\many{\sub{M}{u}};\eta))} \\
    &= \semarr{\vec \U}{\A}{\B}.
  \end{align*}

  Our assumption from Lemma \ref{lem:fun-prefp} is that $c \in \semarr{\vec \U}{\X}{\B}$.
  Moreover, since $X \notin \FV(\Gamma)$, $\sem{\Gamma}(\theta;\eta,\sub{\X}{X}) = \sem{\Gamma}(\theta;\eta) \ni \sigma$.
  Hence $\sigma, \sub{c}{f} \in \sem{\Gamma, f \ann \iarr{\many{u \ann U}}{X \App \vec u}{S}}(\theta;\eta,\sub{\X}{X})$.
  Now we can apply the induction hypothesis with $\eta\pr = \eta, \sub{\X}{X}$ and $\sigma\pr = \sigma, \sub{c}{f}$
  to learn that $s[\theta;\sigma\pr] \evto c\pr$ where
  $c\pr \in \sem{\iarr{\many{u \ann U}}{R[\many{\sub{u}{u\pr}}]}{S}}(\theta;\eta\pr) = \semarr{\vec \U}{\A}{\B}$.


\paragraph*{Case:}
  \[ \infer[\tcorec]{\Delta;\Xi;\Gamma \gives \corec f s \chk \iarr{\many{u \ann U}}{S}{(\corect{X \ann K} \Lam{\many{u\pr}} R) \App \vec{u}}}
                  {\Delta;\Xi,X \ann K;\Gamma, f \ann \iarr{\many{u \ann U}}{S}{X \App \vec u} \gives
                     s \chk \iarr{\many{u \ann U}}{S}{R[\many{u/u\pr}]}
                   & \vec u \notin \FV(R)} \]

  Let $c = (\corec f s)[\theta;\sigma]$ and
  $C = (\corect{X \ann K} \Lam{\many{u\pr}} R) \App \vec{u}$.
  By $\ecorec$, $(\corec f s)[\theta;\sigma] \evto c$.
  We need to show that $c \in \sem{\iarr{\many{u \ann U}}{S}{C}}(\theta;\eta)$,

  Let $\vec \U = \sem{(\many{u \ann U})}(\theta) = \sem{(\many{u\pr \ann U})}(\theta)$
  (both $\vec u$ and $\many{u\pr}$ do not appear in $\theta$).
  From the kinding rules we know that $K = \pik{\many{u\pr \ann U}} \type$ so
  $\sem{K}(\theta) = \vec \U \to \powset \VAL$.
  Let $\mc L = \sem{K}(\theta)$ and define $\F \in \mc L \to \mc L$ by
  $\F(\X) = \outop{\nu} (\sem{\Lam{\many{u\pr}} R}(\theta;\eta,\sub{\X}{X}))$.

  For $\vec M \in \vec \U$,
  \begin{align*}
    \sem{C}(\theta,\many{\sub{M}{u}};\eta)
    &= \sem{\corect{X \ann K} \Lam{\many{u\pr}} R}(\theta,\many{\sub{M}{u}};\eta)(\vec u[\theta,\many{\sub{M}{u}}])
      &\qquad\text{by \sem{T} def} \\
    &= \sem{\corect{X \ann K} \Lam{\many{u\pr}} R}(\theta;\eta)(\vec M) \\
      &\qquad\text{dropping mappings for $\vec u$ since $\vec u \notin \FV(R)$} \\
    &= \semgfp(\X \mapsto \outop{\nu} (\sem{\Lam{\many{u\pr}} R}(\theta;\eta,\sub{\X}{X})))(\vec M)
      &\qquad\text{by \sem{T} def} \\
    &= (\semgfp \F)(\vec M) &\qquad\text{by $\F$ def.} \\
  \end{align*}

  Define $\A \in \mc L$ by $\A(\vec M) = \sem{S}(\theta,\many{\sub{M}{u}};\eta)$.
  Then
  \[ \sem{\iarr{\many{u \ann U}}{S}{C}}(\theta;\eta) = \semarr{\vec \U}{A}{\semgfp \F} \qquad\text{by $\sem T$ def.} \]

  We want to show $c \in \semarr{\vec \U}{\A}{\semgfp \F}$.
  We will instead prove the sufficient condition given in Lemma \ref{lem:fun-prefp}.
  To this end, suppose $\X \in \vec \U \to \powset \VAL = \mc L$ and $c \in \semarr{\vec \U}{\A}{\X}$.
  The goal is now to show $c \in \semarr{\vec \U}{\A}{\F(\X)}$.

  Define $\B \in \mc L$ by $\B = \sem{\Lam{\many{u\pr}} R}(\theta;\eta,\sub{\X}{X})$,
  so $\F(\X) = \outop{\nu} \B$.
  By Lemma \ref{lem:back_clos}, it suffices to show that
  \[ s[\theta;\sigma,\sub{c}{f}] \evto c\pr \in \semarr{\vec \U}{\A}{\B}. \]

  We need to interpret the types $\iarr{\many{u \ann U}}{S}{X \App \vec u}$ and
  $\iarr{\many{u \ann U}}{S}{R[\many{\sub{u}{u\pr}}]}$ appearing in the premise of $\tcorec$.
  Note that these types are well-kinded under the contexts $\Delta;\Xi,X \ann K$.
  Since $\X \in \mc L = \sem{K}(\theta)$, we interpret them under the environments $\theta$ and
  $\eta,\sub{\X}{X} \in \sem{\Xi,X \ann K}(\theta)$.

  \begin{align*}
    &\sem{\iarr{\many{u \ann U}}{S}{X \App \vec u}}(\theta;\eta,\sub{\X}{X}) \\
    &= \semarr{\vec \U}{(\vec M \mapsto \sem{S}(\theta,\many{\sub{M}{u}};\eta,\sub{\X}{X}))}
                             {(\vec M \mapsto \sem{X \App \vec u}(\theta,\many{\sub{M}{u}};\eta,\sub{\X}{X}))} \\
    &= \semarr{\vec \U}{(\vec M \mapsto \sem{S}(\theta,\many{\sub{M}{u}};\eta))}
                             {(\vec M \mapsto \X(\vec M))} \\
    &\qquad\text{dropping a mapping for $X$ since $X \notin \FV(S)$} \\
    &= \semarr{\vec \U}{\A}{\X}
  \end{align*}

  \begin{align*}
    &\sem{\iarr{\many{u \ann U}}{S}{R[\many{\sub{u}{u\pr}}]}}(\theta;\eta,\X/X) \\
    &= \semarr{\vec \U}{(\vec M \mapsto \sem{S}(\theta,\many{\sub{M}{u}};\eta,\sub{\X}{X}))}
                             {(\vec M \mapsto \sem{R[\many{\sub{u}{u\pr}}]}(\theta,\many{\sub{M}{u}};\eta,\sub{\X}{X}))} \\
    &= \semarr{\vec \U}{(\vec M \mapsto \sem{S}(\theta,\many{\sub{M}{u}};\eta))}
                             {(\vec M \mapsto \sem{R}((\id_{\Delta},\many{\sub{u}{u\pr}})[\theta,\many{\sub{M}{u\pr}}];\eta,\sub{\X}{X}))} \\
    &\qquad\text{by Lemma \ref{lem:subst} and again dropping a mapping for $X$} \\
    &= \semarr{\vec \U}{(\vec M \mapsto \sem{S}(\theta,\many{\sub{M}{u}};\eta))}
                             {(\vec M \mapsto \sem{R}(\theta,\many{\sub{M}{u\pr}};\eta,\sub{\X}{X}))} \\
    &= \semarr{\vec \U}{\A}{\B}.
  \end{align*}

  Our assumption from Lemma \ref{lem:fun-prefp} is that $c \in \semarr{\vec \U}{\A}{\X}$.
  Moreover, since $X \notin \FV(\Gamma)$, $\sem{\Gamma}(\theta;\eta,\sub{\X}{X}) = \sem{\Gamma}(\theta;\eta) \ni \sigma$.
  Hence $\sigma, \sub{c}{f} \in \sem{\Gamma, f \ann \iarr{\many{u \ann U}}{S}{X \App \vec u}}(\theta;\eta,\sub{\X}{X})$.
  Now we can apply the induction hypothesis with $\eta\pr = \eta, \sub{\X}{X}$ and $\sigma\pr = \sigma, \sub{c}{f}$
  to learn that $s[\theta;\sigma\pr] \evto c\pr$ where
  $c\pr \in \sem{\iarr{\many{u \ann U}}{S}{R[\many{\sub{u}{u\pr}}]}}(\theta;\eta\pr) = \semarr{\vec \U}{\A}{\B}$.

\paragraph*{Case:}
  \[
\infer{\Delta;\Xi;\Gamma \gives \out{\nu} t \syn T[\many{\sub{M}{u}};\sub{\corect{X \ann K} \Lam{\vec u} T}{X}]}
      {\Delta;\Xi;\Gamma \gives t \syn (\corect{X \ann K} \Lam{\vec u} T) \App \vec M}
      \]

$t[\theta;\sigma] \evto v \in \sem{(\corect{X \ann K} \Lam{\vec u} T) \App \vec M}(\theta;\eta)$ for some $v$ \hfill by I.H.\\
$v \in \sem{\corect{X \ann K} \Lam{\vec u} T}(\theta;\eta)\ (\vec M[\theta])$ \hfill by \sem{T} def.\\
$v \in \semgfp(\X \mapsto \outop{\nu} (\sem{\Lam{\vec u}T}(\theta;\eta,\sub{\X}{X})\}))\ (\vec M[\theta])$ \hfill by $\sem{T}$ def.\\
$v \in \join \{ \C \in \mc L \sep \fall{\X \in \mc L} \C \leq_{\mc L} \X \implies \C \leq_{\mc L}  \outop{\nu} (\sem{\Lam{\vec u}T}(\theta;\eta,\sub{\X}{X}))\}\ (\vec M[\theta])$
\\\hfill by definition of $\semgfp$\\
$v \in \C(\vec M[\theta])$ such that $\fall{\X} \C(\vec M[\theta]) \leq \X(\vec M[\theta])$\\\indent
$\implies \C(\vec M[\theta]) \leq (\outop{\nu} (\sem{\Lam{\vec u}T}(\theta;\eta,\sub{\X}{X})))(\vec M[\theta])$ \hfill by def of $\join$\\
The right-hand side can be simplified as $\C(M[\theta]) \leq (\outop{\nu} (\sem{T}(\theta,\many{\sub{(M[\theta])}{u}};\eta,\sub{\X}{X})))$ \hfill by def of $\sem{T}$\\
Choosing $\sem{\corect{X \ann K} \Lam{\vec u} T}(\theta;\eta)$ for $\X$, the left-hand side
holds trivially by def. of $\sem{\corect{X \ann K} \Lam{\vec u} T}$ and of $\join$\\
Hence, $v \in (\outop{\nu} (\sem{T}(\theta,\many{\sub{(M[\theta])}{u}};\eta,\sub{\sem{\corect{X \ann K} \Lam{\vec u} T}(\theta;\eta)}{X})))$\\
$ \appout v \evto w \in \sem{T}(\theta,\many{\sub{(M[\theta])}{u}};\eta,\sem{\corect{X \ann K} \Lam{\vec u} T}(\theta;\eta)/X)$ \\
$v = (\corec f t)[\theta',\sigma'] \cdot \vec N\ v'$ \hfill by inversion on $\ecorecout$\\
$ (\out{\nu} t)[\sigma;\theta]\evto w \in \sem{T}(\theta,\many{\sub{(M[\theta])}{u}};\eta,\sem{\corect{X \ann K} \Lam{\vec u} T}(\theta;\eta)/X)$ \hfill by $\eout{\nu}$ \\
$\out{\nu} t \evto w \in \sem{T[\many{\sub{M}{u}};(\corect {X{:}K}\Lam {\vec u} T)/X]}(\theta;\eta)$ \hfill  by Lemma \ref{lem:subst}


  \paragraph*{Case:}
  \[ \infer[\tinj{0}]{\Delta;\Xi;\Gamma \gives \inj{0} s \chk T_{\Rec} \App 0 \App \vec M}
                     {\Delta;\Xi;\Gamma \gives s \chk T_z \App \vec M} \]
  $s[\theta;\sigma] \evto w \in \sem{T_z \App \vec M}(\theta;\eta)$ \hfill by I.H. \\
  $(\inj{0} s)[\theta;\sigma] \evto \inj{0} w$ \hfill by $\einj{0}$ \\
  $\inj{0} w \in \inj{0} \sem{T_z \App \vec M}(\theta;\eta)$ \\
  $\inj{0} w \in \sem{T_{\Rec} \App 0 \App \vec M}(\theta;\eta)$ \hfill by Lemma \ref{lem:trec} \\

  \paragraph*{Case:}
  \[ \infer[\tinj{\suc}]{\Delta;\Xi;\Gamma \gives \inj{\suc} s \chk T_{\Rec} \App (\suc N) \App \vec M}
                        {\Delta;\Xi;\Gamma \gives s \chk T_s[\sub{N}{u};\sub{(T_{\Rec} \App N)}{X}] \App \vec M} \]
  $s[\theta;\sigma] \evto w \in \sem{T_s[\sub{N}{u};\sub{(T_{\Rec} \App N)}{X}] \App \vec M}(\theta;\eta)$ \hfill by I.H. \\
  $(\inj{\suc} s)[\theta;\sigma] \evto \inj{\suc} w$ \hfill by $\einj{\suc}$ \\
  $\inj{\suc} w \in \inj{\suc} (\sem{T_s[\sub{N}{u};\sub{(T_{\Rec} \App N)}{X}] \App \vec M}(\theta;\eta))$ \\
  $\inj{\suc} w \in \sem{T_{\Rec} \App (\suc N) \App \vec M}(\theta;\eta)$ \hfill by Lemma \ref{lem:trec} \\

  \paragraph*{Case:}
  \[ \infer[\tout{0}]{\Delta;\Xi;\Gamma \gives \out{0} s \syn T_z \App \vec M}
                     {\Delta;\Xi;\Gamma \gives s \syn T_{\Rec} \App 0 \App \vec M} \]
  $s[\theta;\sigma] \evto w \in \sem{T_{\Rec} \App 0 \App \vec M}(\theta;\eta)$ \hfill by I.H. \\
  $w = \inj{0} v$ for some $v \in (\sem{T_z \App \vec M}(\theta;\eta))$ \hfill by Lemma \ref{lem:trec} \\
  $(\out{0} s)[\theta;\sigma] \evto v$ \hfill by $\eout{0}$ \\

  \paragraph*{Case:}
  \[ \infer[\tout{\suc}]{\Delta;\Xi;\Gamma \gives \out{\suc} s \syn T_s[\sub{N}{u};\sub{(T_{\Rec} \App N)}{X}] \App \vec M}
                        {\Delta;\Xi;\Gamma \gives s \syn T_{\Rec} \App (\suc N) \App \vec M} \]
  $s[\theta;\sigma] \evto w \in \sem{T_{\Rec} \App (\suc N) \App \vec M}(\theta;\eta)$ \hfill by I.H. \\
  $w = \inj{\suc} v$ for some $v \in \sem{T_s[\sub{N}{u};\sub{(T_{\Rec} \App N)}{X}] \App \vec M}(\theta;\eta)$ \hfill by Lemma \ref{lem:trec} \\
  $(\out{\suc} s)[\theta;\sigma] \evto v$ \hfill by $\eout{\suc}$ \\

  \paragraph*{Case:}
  \[ \infer[\tind]{\Delta;\Xi;\Gamma \gives \indnat{t_z}{u}{x}{t_s} \chk \iarr{u \ann \nat}{\Unit}{S}}
  {\Delta;\Xi;\Gamma \gives t_z \chk S[\sub{0}{u}] & \Delta,u \ann \nat;\Xi;\Gamma,x \ann S \gives t_s \chk S[\sub{\suc u}{u}]} \]
  Let $c$ be the closure $(\indnat{t_z}{u}{x}{t_s})[\theta;\sigma]$. \\
  $(\ilam{\vec u}{x}{s})[\theta;\sigma] \evto c$ \hfill by $\eind$ \\
  Suffices to show $c \in \sem{\iarr{u \ann \nat}{1}{S}}(\theta;\eta)$, i.e. \\
  $\fall{N \in \sem \nat} \evalapp{c}{N}{\unit}{w} \in \sem{S}(\theta,\sub{N}{u};\eta)$. \\
  Proceed by induction on $N$. \\
  Base case: $N = 0$. \\
  $t_z[\theta;\sigma] \evto w \in \sem{S[\sub{0}{u}]}(\theta;\eta)$ \hfill by I.H. \\
  $w \in \sem{S}((\id_{\Delta},\sub{0}{u})[\theta];\eta)$ \hfill by Lemma \ref{lem:subst} \\
  $w \in \sem{S}(\theta,\sub{0}{u};\eta)$
  \hfill by \LONGSHORT{Def. \ref{def:isubst-comp}}{composition of index substitutions} \\
  $\appval{c}{0}{\unit} \evto w$ \hfill by $\eappind{0}$ \\
  Step case: $N = \suc N\pr$ for some $N\pr \in \sem \nat$. \\
  $\appval{c}{N\pr}{\unit} \evto v \in \sem{S}(\theta,\sub{N\pr}{u};\eta)$ \hfill by inner I.H. \\
  Let $\theta\pr = \theta,\sub{N\pr}{u}$, so $\gives \theta\pr \oft \Delta,u \ann \nat$. \\
  $\sigma \in \sem{\Gamma}(\theta\pr;\eta)$ \hfill since $u \notin \FV(\Gamma)$ \\
  $\sigma,\sub{v}{x} \in \sem{\Gamma,x \ann S}(\theta\pr;\eta)$ \hfill by Def. \ref{def:sem_typctx} \\
  $t_s[\theta\pr;\sigma,\sub{v}{x}] \evto w \in \sem{S[\sub{\suc u}{u}]}(\theta\pr;\eta)$ \hfill by I.H. \\
  $w \in \sem{S}((\id_{\Delta},\sub{\suc u}{u})[\theta\pr];\eta)$ \hfill by Lemma \ref{lem:subst} \\
  $w \in \sem{S}(\theta\pr,\sub{(\suc u)[\theta\pr]}{u};\eta)$
  \hfill by \LONGSHORT{Def. \ref{def:isubst-comp}}{composition of index substitutions} \\
  $w \in \sem{S}(\theta\pr,\sub{\suc N\pr}{u};\eta)$ \hfill by $\theta\pr$ def. \\
  $w \in \sem{S}(\theta,\sub{N}{u};\eta)$ \hfill by $N$ def. and overwriting $u$ in $\theta\pr$ \\
  $\appval{c}{N}{\unit} \evto w$ \hfill by $\eappind{\suc}$ \\

\end{proof}
}

\end{document}

%% file: spine.tex
\[
  \begin{array}{c}
\multicolumn{1}{l}{
    \fbox{$\Delta \gives (\many{u \ann U}) \itype$} \qquad \mbox{Spine $\many{u \ann U}$ of index types is well-kinded}}\\[0.5em]
\infer{\Delta \gives (\empt) \itype}{} \quad
\infer{\Delta \gives {(u_0{:}U_0, \many{u:U})} \itype}
{\Delta \gives U_0 \itype &
 \Delta, u_0{:}U_0 \gives (\many{u:U}) \itype}
\\[1em]
\multicolumn{1}{l}{
  \fbox{$\Delta \gives \vec M \oft (\many{u \ann U})$}  \qquad \mbox{Spine $\vec M$ of index terms have index types $(\many{u \ann U})$}}\\[0.5em]
\infer{\Delta \gives \empt : (\empt)}{}
\qquad
\infer{\Delta \gives M_0, \vec M \oft (u_0{:}U_0, \many{u{:}U})}
{\Delta \gives M_0 \oft U_0 &
    \Delta \gives \vec M \oft (\many{u{:}U})[\sub{M_0}{u_0}] }
  \end{array}
\]

%% file: subred.tex
\\[1em]
 \textbf{Case}~~ $\icnc{\eval{s}{\theta;\sigma}{\refl}}{\Delta\pr \gives \Theta \match \theta \synth (\empt \sep \theta\pr)}{\eval{t}{\theta\pr;\sigma}{v}}{\eval{(\eqelim{s}{\Theta}{\Delta\pr}{t})}{\theta;\sigma} v}{}$
\\[1em]
$\Delta;\empt;\Gamma \gives s \syn M = N$ and
$\Delta\pr;\empt;\Gamma[\Theta] \gives t \chk T[\Theta]$ \hfill by inversion of typing \\
$\refl \oft (M = N)[\theta]$ \hfill by I.H. \\
$\refl \oft M[\theta] = N[\theta]$ \hfill by type substitution \\
$\gives M[\theta] = N[\theta]$ \hfill by inversion of value typing \\
$\gives \theta\pr \oft \Delta\pr$ and $\gives \Theta[\theta\pr] = \theta$ \hfill by soundness of matching (Req. \ref{thm:sound-match}) \\
$T[\Theta][\theta\pr] = T[\Theta[\theta\pr]] = T[\theta]$ \hfill by associativity of type substitution \\
$\Gamma[\Theta][\theta\pr] = \Gamma[\Theta[\theta\pr]] = \Gamma[\theta]$ \hfill similarly for contexts \\
$\sigma \oft \Gamma[\theta]$ \hfill by assumption \\
$\sigma \oft \Gamma[\Theta][\theta\pr]$ \hfill by context equality \\
$v \oft T[\Theta][\theta\pr]$ \hfill by I.H. \\
$v \oft T[\theta]$ \hfill by type equality
\\[1em]
 \textbf{Case}~~ $\icnc{\eval{t}{\theta;\sigma}{c}}{\eval{s}{\theta;\sigma}{v}}{\appval{c}{\many{M[\theta]}}{v} \evto w}{\eval{(t \app \vec M \app s)}{\theta;\sigma}{w}}{}$
\\[1em]
$\gives \theta : \Delta$ and $\sigma : \Gamma[\theta]$ \hfill by assumption \\
$\Delta;\cdot;\Gamma \gives t \app \vec M \app s \syn T[\many{M/u}]$ \hfill by assumption \\
$\Delta;\cdot;\Gamma \gives t \syn \iarr{\many{u \ann U}}{S}{T}$ \\
$\Delta \gives \vec M \oft (\many{u \ann U})$ \\
$\Delta;\cdot;\Gamma \gives s \chk S[\many{\sub{M}{u}}]$ \hfill by inversion of typing \\
$\gives \many{M[\theta]} \oft (\many{u \ann U})[\theta]$ \hfill extending Req. \ref{req:isubst-typ-sim} to index spines \\
$c : (\iarr{\many{u \ann U}}{S}{T})[\theta]$ \hfill by I.H. \\
$v : S[\many{M/u}][\theta]$ \hfill by I.H. \\
$v : S[\theta, \many{\sub{M[\theta]}{u}}]$ \hfill by associativity of type substitution \\
$c = g[\theta\pr; \sigma\pr]$ where $g$ is a function value
\hfill by closure grammar and typing \\
$\Delta\pr; \empt; \Gamma\pr \gives g \chk G$ and $g[\theta\pr ; \sigma\pr] : G[\theta\pr]$
\hfill by closure typing \\
$G[\theta\pr] = (\iarr{\many{u \ann U}}{S}{T})[\theta]$ \hfill by previous lines \\
$G = \iarr{\many{u\ann U\pr}}{S\pr}{T\pr}$ where $\many{U\pr[\theta\pr]} = \many{U[\theta]}$ \\
and $S\pr[\theta\pr, \many{\sub{u}{u}}] = S[\theta, \many{\sub{u}{u}}]$
and $T\pr[\theta\pr, \many{\sub{u}{u}}] = T[\theta, \many{\sub{u}{u}}]$ \hfill by equality of types \\
$v : S\pr[\theta\pr, \many{\sub{M[\theta]}{u}}]$ \hfill by type equality \\
$\gives \many{M[\theta]} : (\many{u\ann U\pr[\theta\pr]})$ \hfill by type equality \\
$w : T\pr[\theta\pr, \many{\sub{M[\theta]}{u}}]$ \hfill by I.H. 2 \\
$w : T[\theta, \many{\sub{M[\theta]}{u}}]$ \hfill by type equality
\\[1em]
\textbf{Case}~~$\ianc{t[\theta,\many{\sub{N}{u}};\sigma,\sub{v}{x}] \evto w}
                     {\evalapp{(\ilam{\vec u}{x} t)[\theta;\sigma]}{\vec N}{v}{w}}{}$
\\[1em]
$\Delta ; \cdot ; \Gamma \gives \ilam{\vec u}{x} t \chk \iarr{\many{u\ann U}}{S}{T}$ \hfill by assumption \\
$\Delta,\many{u \ann U};\dot;\Gamma, x \ann S \gives t \chk T$ \hfill by inversion of typing \\
$\gives \theta : \Delta$ \hfill by assumption \\
$\gives \theta, \many{\sub{N}{u}} : \Delta, \many{u\ann U}$ \hfill by substitution typing \\
$\sigma : \Gamma[\theta]$ \hfill by assumption \\
$\sigma : \Gamma[\theta, \many{\sub{N}{u}}]$ \hfill by weakening since $\vec u$ do not occur in $\Gamma$ \\
$\sigma, \sub{v}{x} : \Gamma[\theta, \many{\sub{N}{u}}], x:S[\theta, \many{\sub{N}{u}}]$ \hfill by environment typing \\
$\sigma, \sub{v}{x} : (\Gamma, x:S)[\theta, \many{\sub{N}{u}}]$ \hfill by def. of context substitution \\
$w : T[\theta, \many{\sub{N}{u}}]$ \hfill by I.H.
\\[1em]
\textbf{Case}~~$\ibnc{t[\theta;\sigma,\sub{(\rec f t)[\theta;\sigma]}{f}] \evto c}
                     {\evalapp{c}{\vec N}{v\pr}{w}}
                     {\evalapp{(\rec f t)[\theta;\sigma]}{\vec N}{(\inj{\mu} v\pr)}{w}}{}$
\\[1em]
$\Delta;\cdot;\Gamma \gives \rec f t \chk \iarr{\many{u \ann U}}{(\rect{X \ann K} \Lam{\many{v}} S) \App \vec{u}}{T}$ \hfill by assumption
\\
$\Delta;X \ann K;\Gamma,f \ann (\iarr{\many{u \ann U}}{X \vec{u}}{T})
  \gives t \chk \iarr{\many{u \ann U}}{S[\many{\sub{u}{v}}]}{T}$ \hfill by inversion of typing \\
$\Delta; \empt; \Gamma , f \ann (\iarr{\many{u \ann U}}{(\rect{X \ann K} \Lam{\many{v}} S)~\vec u}{T}) \gives t \chk
\iarr{\many{u \ann U}}{S[\many{\sub{u}{v}}; \rect{X \ann K} \Lam{\many{v}} S/X]}{T}$ \\
  {} \hfill by substitution property of type variables \\
$\inj{\mu}v\pr : ((\rect{X \ann K} \Lam{\many{v}} S)~\vec u))[\theta, \many{N/u}]$ \hfill by assumption \\
$v\pr : S[\theta, \many{N/u};~\rect{X \ann K} \Lam{\many{v}} S/X]$ \hfill by value typing \\
$\gives \theta : \Delta$ \hfill by assumption \\
$\sigma, \sub{(\rec f t)[\theta;\sigma]}{f} : \Gamma[\theta], f \ann (\iarr{\many{u \ann U}}{(\rect{X \ann K} \Lam{\many{v}} S)~\vec u}{T})[\theta]$
\hfill by environment typing \\
$\sigma, \sub{(\rec f t)[\theta;\sigma]}{f} : (\Gamma, f \ann (\iarr{\many{u \ann U}}{(\rect{X \ann K} \Lam{\many{v}} S)~\vec u}{T}))[\theta]$
\hfill by def. of context substitution \\
$c : (\iarr{\many{u \ann U}}{S[\many{\sub{u}{v}} ;\rect{X \ann K} \Lam{\many{v}} S/X]}{T})[\theta]$ \hfill by I.H.  \\
$c = g[\theta\pr; \sigma\pr]$ where $\Delta\pr; \empt; \Gamma\pr \gives g \chk G$ and $g[\theta\pr ; \sigma\pr] : G[\theta\pr]$ \\
and $\gives \theta\pr : \Delta\pr$ and $\sigma\pr : \Gamma\pr[\theta\pr]$
\hfill by closure grammar and typing \\
$G[\theta\pr] = (\iarr{\many{u \ann U}}{S}{T})[\theta]$ \hfill by previous lines \\
$G = \iarr{\many{u\ann U\pr}}{S\pr}{T\pr}$ where $\many{U\pr[\theta\pr]} = \many{U[\theta]}$ \\
and $S\pr[\theta\pr, \many{\sub{u}{u}}] = S[\theta, \many{\sub{u}{u}}~;~ \rect{X \ann K} \Lam{\many{v}} S/X]$
and $T\pr[\theta\pr, \many{\sub{u}{u}}] = T[\theta, \many{\sub{u}{u}}]$ \hfill by type equality \\
$\Delta\pr; \empt; \Gamma\pr \gives g \chk \iarr{\many{u\ann U\pr}}{S\pr}{T\pr}$ \hfill by previous lines \\
$\gives \vec N \oft (\many{u:U[\theta]})$ \hfill by assumption \\
$\gives \vec N : (\many{u:U\pr[\theta\pr]})$ \hfill by type equality \\
$v\pr : S\pr[\theta\pr, \many{\sub{N}{u}}]$ \hfill by type equality \\
$w : T\pr[\theta\pr,\many{\sub{N}{u}}]$ \hfill by I.H. \\
$w : T[\theta, \many{\sub{N}{u}}]$ \hfill by type equality \\
\\[1em]
\textbf{Case}~~$\infer{\appout{((\corec f t)[\theta;\sigma] \cdot \vec N\ v)} \evto w}
      {t[\theta;\sigma,\sub{(\corec f t)[\theta;\sigma]}{f}] \evto c
        & \evalapp{c}{\vec N}{v}{w}}$
      \\[1em]
$(\corec f t)[\theta;\sigma] \cdot \vec N\ v
      : (\corect{X \ann K} \Lam{\many{u\pr}} T) \App \vec M$ \hfill by assumption \\
$\close{\corec f t}{\theta}{\sigma} \cdot \vec N\ v
      : \msub \theta {(\corect{X \ann K} \Lam{\many{u\pr}} T\pr)}\ \vec N$ \hfill by value typing \\
$(\corect{X \ann K} \Lam{\many{u\pr}} T) = (\corect{X \ann K} \Lam{\many{u\pr}} T\pr)[\theta]$ and $\vec N = \vec M$ \hfill by type equality \\
$\Delta;\cdot,X \ann K;\Gamma,f \ann (\iarr{\many{u \ann U}}{S}{X \vec{u}})
      \gives t \chk \iarr{\many{u \ann U}}{S}{T'[\many{\sub{u}{u\pr}}]}$ \\
and $\cdot \gives \theta : \Delta$ \\
and $\sigma : \msub{\theta} \Gamma$ \\
and  $\cdot \vdash \vec N : \vec U$ \\
and $v : S[\theta,\many{\sub{N}{u}}]$ \hfill by inversion on value typing \\
$\Delta;\cdot;\Gamma,f \ann (\iarr{\many{u \ann U}}{S}{(\corect{X \ann K} \Lam{\many{u\pr}} T\pr)\ \vec{u}})
     \gives t' \chk \iarr{\many{u \ann U}}{S}{T\pr[\many{\sub{u}{u\pr}};(\corect{X \ann K} \Lam{\many{u\pr}} T\pr)/X]}$
     \\ {} \hfill by substitution property of type variables \\
$\sigma, \sub{(\corec f t)[\theta;\sigma]}{f} : \Gamma[\theta], f \ann (\iarr{\many{u \ann U}}{S}{(\corect{X \ann K} \Lam{\many{u\pr}} T\pr)~\vec u})[\theta]$
\hfill by environment typing \\
$\sigma, \sub{(\corec f t)[\theta;\sigma]}{f} : (\Gamma, f \ann (\iarr{\many{u \ann U}}{S}{(\corect{X \ann K} \Lam{\many{u\pr}} T\pr)~\vec u}))[\theta]$
\hfill by def. of ctx subst. \\
$c : (\iarr{\many{u \ann U}}{S}{T\pr[\many{\sub{u}{u\pr}};(\corect{X \ann K} \Lam{\many{u\pr}} T\pr)/X]})[\theta]$ \hfill by I.H.\\
$c = g[\theta';\sigma']$ where $\Delta'; \cdot; \Gamma' \vdash g \chk G$
and $g[\theta';\sigma'] : G[\theta']$ \hfill by closure grammar and typing\\
$G[\theta'] = (\iarr{\many{u \ann U}}{S}{T\pr[\many{\sub{u}{u\pr}};(\corect{X \ann K} \Lam{\many{u\pr}} T\pr)/X]})[\theta]$
\hfill by type uniqueness\\
$G = (\iarr{\many{u \ann U\pr}}{S\pr}{T''})$ where $\many{U\pr[\theta']} = \many{U[\theta]}$
and $S\pr[\theta', \many{\sub{u}{u}}] = S[\theta, \many{\sub{u}{u}}]$\\
and $T''[\theta', \many{\sub{u}{u}}] = T\pr[\theta, \many{\sub{u}{u\pr}}~;~ (\corect{X \ann K} \Lam{\many{u\pr}} T\pr)[\theta]/X]$ \hfill by type equality \\
$\Delta'; \empt; \Gamma' \gives g \chk \iarr{\many{u\ann U\pr}}{S\pr}{T''}$ \hfill by previous lines \\
$\vdash \vec N : \many{U\pr[\theta']}$ \hfill by type equality \\
$v : S'[\theta', \many{\sub{N}{u}}]$ \hfill by type equality \\
$w : T''[\theta',\many{\sub{N}{u}}]$ \hfill by I.H. \\
$w : T'[\theta, \many{\sub{M}{u\pr}}~;~ (\corect{X \ann K} \Lam{\many{u\pr}} T')[\theta]/X]$ \hfill by type equality \\
$w : T[\many{\sub{M}{u\pr}}~;~ (\corect{X \ann K} \Lam{\many{u\pr}} T)/X]$ \hfill by type equality
